\documentclass[a4paper]{article}
\usepackage{lmodern}
\usepackage[top=1in, bottom=1in, left=1in, right=1in]{geometry}

\usepackage{amsmath,amssymb,amsthm}
\usepackage{wrapfig}
\usepackage{tikz}
\usetikzlibrary{positioning}

\newtheorem{theorem}{Theorem}
\newtheorem{conjecture}{Conjecture}
\newtheorem{definition}{Definition}
\newtheorem{remark}{Remark}

\newtheorem{lemma}{Lemma}
\newtheorem{fact}[lemma]{Fact}
\newtheorem{corollary}{Corollary}

\newtheorem*{lemma*}{Lemma}
\newtheorem*{corollary*}{Corollary}
\newtheorem*{theorem*}{Theorem}

\usepackage{xcolor}
\usepackage[normalem]{ulem}

\newcommand{\apndeg}{\mathsf{N}}

\usepackage{hyperref}

\author{%
\begin{tabular}{c@{\hspace{2.2em}}c}
Samruddhi Pednekar & Supartha Podder
\end{tabular}\\[0.9em]
{\small Stony Brook University, New York, USA}\\[0.4em]
{\small \texttt{\char`\{spednekar, supartha\char`\}@cs.stonybrook.edu}}
}

\title{On the Approximate Non-Deterministic Degree of Total Boolean Functions}

\begin{document}

\maketitle

\begin{abstract}
The approximate non-deterministic degree of a Boolean function $f$,  denoted $\mathsf{ndeg}_\epsilon(f)$ (written $\mathsf{N}_\epsilon(f)$ for brevity), is the minimum degree of a real polynomial $p$ such  that $0 \le |p(x)| \le \epsilon$ whenever $f(x) = 0$, and $|p(x)| \ge 1$  whenever $f(x) = 1$. Unlike exact non-deterministic degree, which only  requires the polynomial to be nonzero on $1$-inputs, this measure  enforces a uniform gap: the polynomial must stay close to zero on all  $0$-inputs and bounded away from zero on all $1$-inputs.

The rational degree conjecture, open for over three decades, was recently resolved by Kothari, Kovacs-Deak, Wang, and Yang~\cite{kovacsdeak2026rationaldegreepolynomiallyrelated}, who showed that for every total Boolean function $f$,
\[ \deg(f) \le \widetilde O\!\left(\operatorname{rdeg}(f)^3\right). \]

In their paper, they explicitly propose a stronger conjecture: that approximate degree is polynomially bounded by $\mathsf{N}_\epsilon(f)$ and $\mathsf{N}_\epsilon(\overline{f})$ jointly, i.e., for every total Boolean function $f$ and every constant
$0<\epsilon<1$,
\[ \widetilde{\deg}(f) \le \operatorname{poly}(\mathsf N_{\epsilon}(f), \mathsf N_{\epsilon}(\overline f)). \]

This conjecture, if true,would imply a polynomial version of the rational degree result and bring us closer to resolving de Wolf's longstanding non-deterministic degree conjecture~\cite{dewolf2004nondeterministicquantumqueryquantum}.

In this work, we make the first systematic progress on this problem, establishing the conjecture for several broad and natural function classes: monotone and unate functions, functions of bounded alternation number, symmetric functions, $k$-uniform hypergraph properties, and read-$k$ Disjunctive Normal Form (DNF) formulas.
\end{abstract}

\newpage

\section{Introduction}
\label{sec:typesetting-summary}

A central goal in the analysis of Boolean functions is to understand how
different complexity measures relate to one another.  Many of these measures come from very different models: decision trees, certificates, block sensitivity, quantum query algorithms, and real polynomial representations. Nevertheless, for total Boolean functions, many of them are known to be polynomially related.  The foundational work of Nisan and
Szegedy~\cite{nisan_szegedy} showed that the degree and approximate degree of a Boolean function are polynomially related to several standard complexity measures, including block sensitivity, certificate complexity, and deterministic query complexity.  This circle of equivalences was later extended to other measures, and one of the main remaining gaps, the sensitivity conjecture, was resolved by Huang~\cite{huang2019induced}.

Another major open problem from the same line of work concerned rational
representations of Boolean functions. This question, attributed to Fortnow and recorded by Nisan and Szegedy~\cite{nisan_szegedy}, asked whether rational degree is polynomially related to ordinary degree.  The rational degree of a Boolean function $f:\{0,1\}^n\to\{0,1\}$, denoted $\operatorname{rdeg}(f)$, is the least degree needed to represent $f$ exactly on the Boolean cube as a quotient $p/q$ of real polynomials, where $q$ is nonzero on every input.  Since every polynomial representation is also a rational representation with denominator $1$, we always have
\[ \operatorname{rdeg}(f)\le \deg(f). \]
The rational degree conjecture asked for a converse up to polynomial loss:
does there exist a constant $c$ such that
\[ \deg(f)\le O(\operatorname{rdeg}(f)^c) \]
for every total Boolean function $f$?

This conjecture remained open for more than three decades.  It was recently proved by Kothari, Kovacs-Deak, Wang, and Yang~\cite{kovacsdeak2026rationaldegreepolynomiallyrelated}, who showed that, for every total Boolean function $f$,
\[ \deg(f) \le \widetilde O\!\left(\operatorname{rdeg}(f)^3\right). \]
Thus rational degree also belongs to the family of Boolean function measures that are polynomially related for total functions
\cite{kovacsdeak2026rationaldegreepolynomiallyrelated}.  The restriction to total functions is important: for partial functions, the analogous statement is false, as shown by Iyer et al.~\cite{iyer2023rationaldegreebooleanfunctions}.

A key reason rational degree is natural is its close connection to
nondeterminism.  A non-deterministic polynomial for $f$ is a real polynomial
$p$ such that
\[ p(x)\neq 0 \quad\Longleftrightarrow\quad f(x)=1 \qquad \text{for all } x\in\{0,1\}^n . \]
The non-deterministic degree of $f$, denoted $\operatorname{ndeg}(f)$, is
the minimum degree of such a polynomial.  Rational degree can be described
exactly in terms of non-deterministic degree:
\[ \operatorname{rdeg}(f) = \max\{ \operatorname{ndeg}(f), \operatorname{ndeg}(\overline f) \}. \]

In other words, rational degree measures the larger of the two
non-deterministic degrees associated with $f$ and its complement. 
This connection also brings quantum query complexity into the picture.
de Wolf showed that non-deterministic degree is equal, up to constant factors,
to non-deterministic quantum query complexity
\cite{dewolf2004nondeterministicquantumqueryquantum}.  Motivated by the
classical relation
\[ D(f)\le C^{(0)}(f)C^{(1)}(f), \]
or equivalently by the product of non-deterministic query complexities for $f$ and $\overline f$, de Wolf conjectured that
\begin{conjecture}[\cite{dewolf2004nondeterministicquantumqueryquantum}]\label{conjecture1}
    For every total Boolean function $f$,
    \[ D(f) \le O\!\left(\operatorname{ndeg}(f)\operatorname{ndeg}(\overline f) \right) \]
\end{conjecture}
 The dependence on both $f$ and $\overline f$ is necessary: for example, $\operatorname{OR}_n$ has very small non-deterministic degree, while its deterministic query complexity is large, and the same obstruction appears for $\operatorname{AND}_n$ after
taking complements.

de Wolf's conjecture is stronger than the rational degree conjecture.  Indeed,
since $\operatorname{rdeg}(f) = \max\{\operatorname{ndeg}(f), \operatorname{ndeg}(\overline f) \}$,
the conjecture would imply a quadratic upper bound of ordinary degree in terms of rational degree.  Moreover, since $\operatorname{ndeg}(f) \le \deg(f)$ and $\operatorname{ndeg}(\overline{f}) \le \deg(\overline{f})$ and $\deg(f)=\deg(\overline{f})$, the conjecture would also give $D(f) \le O(\deg(f)^2)$. \cite{kovacsdeak2026rationaldegreepolynomiallyrelated} gave a weaker polynomial form.  In particular, they show
\[ D(f) \le O\!\left( \operatorname{ndeg}(f)^2 \operatorname{ndeg}(\overline f)^2 \right),\quad \textit{and \quad}  D(f) \le O\!\left( \operatorname{ndeg}(f)^{1.5} \operatorname{ndeg}(\overline f)^{1.5} \log n \right), \]
but the original product conjecture remains open
\cite{kovacsdeak2026rationaldegreepolynomiallyrelated}.

\subsection{Approximate version of non-deterministic degree}
\cite{dewolf2004nondeterministicquantumqueryquantum, kovacsdeak2026rationaldegreepolynomiallyrelated} proposed an even stronger version of Conjecture~\ref{conjecture1} involving approximate non-deterministic degree: For any constant $0<\epsilon<1$ define the
$\epsilon$-approximate non-deterministic degree of $f$, denoted\footnote{Prior work~\cite{dewolf2004nondeterministicquantumqueryquantum, kovacsdeak2026rationaldegreepolynomiallyrelated} denotes $\epsilon$-approximate non-deterministic degree of $f$ as  $\operatorname{ndeg}_\epsilon(f)$. We denote this as  $\mathsf{N}_\epsilon(f)$ for easy readability. We also used $\mathsf{N}(f)$ for $\mathsf{N}_\epsilon(f)$ where $\epsilon =1/3$. 
These should not be confused with the notation $\mathsf{N}_q(f)$ or $\mathsf{N}_c(f)$ used in~\cite{dewolf2004nondeterministicquantumqueryquantum} for denoting non-deterministic classical query complexity and non-deterministic classical communication complexity respectively, which we do not discuss here.}
$\mathsf N_{\epsilon}(f)$, to be the minimum degree of a real polynomial
$p$ such that
\[ |p(x)|\le \epsilon \quad\text{whenever } f(x)=0, \qquad |p(x)|\ge 1 \quad\text{whenever } f(x)=1 .\]
When $\epsilon=0$, this recovers the \emph{exact} non-deterministic degree, denoted by $\operatorname{ndeg}(f)$.  For
positive $\epsilon$, however, the measure becomes more delicate.  Exact non-deterministic degree only requires the polynomial 
\begin{wrapfigure}{r}{0.4\textwidth}
    \centering
    \begin{tikzpicture}[scale=0.88, 
        every node/.style={font=\footnotesize},
        >={latex}
    ]
        \node (deg)      at (1.5,    4)   {$\deg(f)$};
        \node (adeg)     at (-0.5, 2) {$\widetilde{\deg}(f)$};
        \node (rdeg)     at (3,    2.5) {$\operatorname{rdeg}(f) = \max\{\operatorname{ndeg}(f),\operatorname{ndeg}(\overline{f})\}$};
        \node (ndeg)     at (3.1,    1.2) {$\operatorname{ndeg}(f)$};
        \node (odeg)     at (2.75,    -0.5)   {$odeg_\epsilon(f) = \Theta( \operatorname{ndeg}\epsilon(f))$};
        \node (ardeg)    at (-0.5,-1.2) {$\widetilde{\operatorname{rdeg}}(f)$};

        \draw[<-] (deg)  -- (rdeg);
        \draw[<-] (deg)  -- (adeg);
        \draw[<-] (rdeg) -- (ndeg);
        \draw[->] (odeg) -- (ndeg);
        \draw[->] (odeg) -- (adeg);
        \draw[<-] (adeg) -- (ardeg);
    \end{tikzpicture}
    \caption{Relations between different degrees for total Boolean functions. $A \rightarrow B$ denotes $A \leq B$.}
    \label{relations}
\end{wrapfigure}
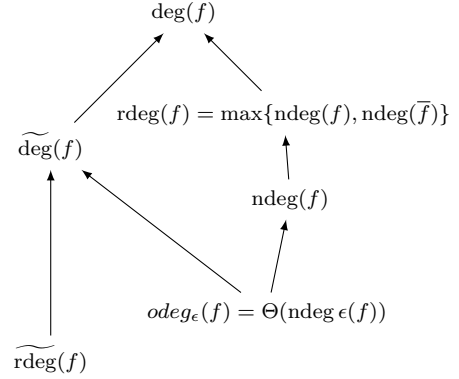
to be nonzero on $1$ inputs.  Approximate non-deterministic degree asks for a uniform gap:
the polynomial must be small on all $0$ inputs and bounded away from zero on all $1$ inputs.  Thus it is a one-sided bounded version of non-deterministic degree.  This notion is closely related, up to changes in constants, to one-sided approximate degree, which has been studied for example in~\cite{SHE18, BT22}.

The natural two-sided analogue is the widely studied measure approximate degree.  Recall that
$\widetilde{\deg}(f)$ is the least degree of a real polynomial that
approximates $f$ pointwise within constant error on the Boolean cube.  The main question we consider is whether low-degree one-sided approximations for both $f$ and $\overline f$ force $f$ to have low approximate degree.

\begin{conjecture}[Approximate non-deterministic degree conjecture -- exact form]
\label{conj:approx-ndeg-D}
For every total Boolean function $f$ and every constant
$0<\epsilon<1$,
\[ \widetilde{\deg}(f) \le \max\{\mathsf N_{\epsilon}(f), \mathsf N_{\epsilon}(\overline f)\}. \]
\end{conjecture}

In fact, a weaker polynomial bound is also unknown~\cite{dewolf2004nondeterministicquantumqueryquantum}.

\begin{conjecture}[Approximate non-deterministic degree conjecture -- polynomial form]
\label{conj:approx-ndeg-Dpoly}\label{approx_ndeg_conjecture}
For every total Boolean function $f$ and every constant
$0<\epsilon<1$,
\[ \widetilde{\deg}(f) \le \operatorname{poly}(\mathsf N_{\epsilon}(f), \mathsf N_{\epsilon}(\overline f)). \]
\end{conjecture}

\subsection{Our results}
Our main positive results are summarized below: We show that  Conjecture~\ref{approx_ndeg_conjecture} is true for the following classes.

\begin{enumerate}
    \item (Corollary~\ref{lem:monotone}) For every \emph{monotone} Boolean function $f:\{0,1\}^n \rightarrow \{0,1\}$,
    \[ \widetilde{\deg}(f) \le O\!\left( \max\{\mathsf N_\epsilon(f), \mathsf N_\epsilon(\overline f)\}^4 \right).\]
    The same bound holds for \emph{unate} functions (Lemma~\ref{lem:unate-functions}).

    \item (Corollary~\ref{cor:bounded-alt-D}) For every Boolean function $f:\{0,1\}^n \rightarrow \{0,1\}$,
    \[ D(f) \le O\!\left( \operatorname{alt}(f)^3 \max\{ \mathsf N_\epsilon(f), \mathsf N_\epsilon(\overline f) \}^6 \right).\]
    where $\operatorname{alt}(f)$ is the alternation number of $f$.     In particular, this gives the desired polynomial bound for functions of bounded alternation.
    \item (Corollary~\ref{zebra_D_approxndeg}) For every \emph{zebra} function $f$,
    \[ D(f) \le O\!\left( \operatorname{alt}(f)^2 \max\{ \mathsf N_\epsilon(f), \mathsf N_\epsilon(\overline f) \}^4 \right).\]

    \item (Corollary~\ref{symmetric_relation}) For every \emph{symmetric} Boolean function $f:\{0,1\}^n \rightarrow \{0,1\}$,
    \[ D(f) \le O\!\left( \max\{ \mathsf N_\epsilon(f), \mathsf N_\epsilon(\overline f) \}^6 \right).\]
    
    For a large subclass of symmetric functions, namely those that are
    constant on a long interval of Hamming weights around the middle layer, we obtain the following stronger bound (Corollary~\ref{cor:symmetric-stronger}),
    \[\widetilde{\deg}(f) \le O\!\left( \max\{ \mathsf N_\epsilon(f), \mathsf N_\epsilon(\overline f) \}^2 \right).\]

    \item (Theorem~\ref{thm:readk-dnf-main}) Let $f$ be any Boolean function that is computed by a read-$k$ Disjunctive Normal Form (DNF) formulas, then
    \[ \widetilde{\deg}(f) \le O\!\left( k(k+1)^2(2k+1)^2 \bigl(\mathsf N_\epsilon(f) \mathsf N_\epsilon(\overline f)\bigr)^6\right).\]
    In particular, for constant $k$, this gives the desired polynomial bound.

    \item (Theorem~\ref{thm:hypergraph}) For any constant $k\geq 1$, and for every Boolean function $f:\{0,1\}^{\binom{n}{k}} \rightarrow \{0,1\}$ representing a \emph{$k$-uniform hypergraph property},
    \[ D(f) \le O\!\left( \max\{ \mathsf N_\epsilon(f), \mathsf N_\epsilon(\overline f) \}^{6k} \right).\]
    where $n$ represents the number of vertices. 
    Consequently, \emph{graph properties} satisfy the conjecture up to a polynomial loss (Corollary~\ref{graph_property_result}).
    
\end{enumerate}

Most of our proofs are written for $\epsilon=1/3$, but the same arguments extend to any fixed constant $0<\epsilon<1$, changing the bounds only by constant factors.

\subsection{Motivation for the function classes}

The classes of functions considered here are well studied in the analysis of Boolean functions and query complexity.
Going back to Aanderaa--Karp--Rosenberg--conjecture and early work on evasiveness~\cite{rivest1976recognizing, yao1987lower, chakrabarti2001improved, dietmar1992randomized, hajnal1991omega, king1988lower,kahn1984topological}, graph properties are well studied. More recently, graph properties have also been studied in the context of quantum query complexity~\cite{aaronson2021degree, kulkarni2016quantum, mariella2023quantum, ben2024symmetries, childs2012quantum, balaji2016graph}.  Hypergraph properties are a natural generalization of graph properties which have also been studied in various previous work~\cite{alon1985hypergraphs, kulkarni2015any, black2015monotone, li2021sensitivity}. 
Functions of bounded alternation were studied in~\cite{lin2016sensitivityconjecturelogrankconjecture, blais2014learning, canonne2019testing, dinesh2019alternation, chugh2022decisiontreecomplexityversus}. Zebra functions were defined as a natural restriction of alternation and as a class of functions that contains symmetric and monotone functions as its
subclasses~\cite{chugh2022decisiontreecomplexityversus}. 
Read-$k$ formulas and DNF formulas have also been studied extensively in
connection with sensitivity and block sensitivity~\cite{bafna2016sensitivity,tavenas2016sensitivity}.

\subsection*{Acknowledgments} We thank Thomas Huffstutler and David Miloschewsky for helpful discussions and suggesting improvements to the presentation of the paper. We also thank Anna G\'al for related helpful discussions. 

\section{Preliminaries}

Throughout the paper, let $0<\epsilon<1$ be a fixed constant. 
For ease of reading, for every Boolean function $f:\{0,1\}^n\to\{0,1\}$, we use the notation $\mathsf{N}_\epsilon(f)$ to denote $\mathsf{ndeg}_{\epsilon}(f)$. When $\epsilon = 1/3$, we drop the $\epsilon$ and use $\mathsf{N}(f)$. Thus, $\mathsf{N}(f)=\mathsf{ndeg}_{1/3}(f)$.
We also define
\[M_\epsilon(f) := \max\{\mathsf{N}_\epsilon(f), \mathsf{N}_\epsilon(\overline f)\}. \]
When $\epsilon=1/3$, we simply write $M(f)$. 

We will use $[n]$ to denote the set $\{1,2,\dots,n\}$. For an input $x \in \{0,1\}^n$, and a subset $B \subseteq [n]$, $x^B$ denotes flipping the bits of $x$ at all indices $i \in B$.
We will denote $|x|$ as the Hamming weight $\sum_i x_i$ for $x \in \{0,1\}^n$. For $x,y \in \{0,1\}^n$, $x \preceq y$ if and only if $x_i \le y_i$ for all $i \in [n]$ which means that $x_i=1 \implies y_i=1$ and $y_i=0 \implies x_i=0$.

\begin{definition}[Degree of Boolean function]
    Let $f:\{0,1\}^n \rightarrow \{0,1\}$ and $p(x)$ is a real polynomial $p(x) = \sum_{S \subseteq [n] }c_s \prod_{i \in S} x_i$ with $p(x) = f(x)$ for all $x \in \{0,1\}^n$. Then, $\deg(p) = \max\left\{|S| \mid c_s \ne 0\right\}$ is called the degree of the Boolean function $f$, denoted by $\deg(f)$.
\end{definition}

\begin{definition}[Non-deterministic degree~\cite{dewolf2004nondeterministicquantumqueryquantum}]\label{ndeg_def}
    Let $f:\{0,1\}^n \rightarrow \{0,1\}$ and polynomial $p$ is called non-deterministic polynomial of $f$ if $p(x)=0$ if and only if $f(x)=0$. $\mathsf{ndeg}(f)$ is the minimum degree of such polynomial $p(x)$.
\end{definition}

\begin{definition}[Approximate non-deterministic degree~\cite{kovacsdeak2026rationaldegreepolynomiallyrelated}]\label{approx_ndeg_def}
    Let $f:\{0,1\}^n \rightarrow \{0,1\}$ and a real polynomial $p$ is called an approximate non-deterministic polynomial of $f$ if $|p(x)|\le \epsilon$ when $f(x)=0$ and $|p(x)|\ge 1$ when $f(x)=1$. $\mathsf{N}_\epsilon(f) = \mathsf{ndeg}_{\epsilon}(f)$ is the minimum degree of such polynomial $p(x)$.
\end{definition}

\begin{definition}[Approximate Rational Degree~\cite{iyer2023rationaldegreebooleanfunctions}]\label{rational_deg_def}
    Let $f:\{0,1\}^n \rightarrow \{0,1\}$ and $p: \mathbb{R}^n \rightarrow \mathbb{R}$ and $q:\mathbb{R}^n \rightarrow \mathbb{R}$ be polynomials such that $\frac{p(x)}{q(x)}$ is a rational approximation of $f(x)$. That is,
    \[ \forall x \in \{0,1\}^n, \left|f(x) - \frac{p(x)}{q(x)}\right| \le \epsilon.\]
    
     The minimum value of $\max\{\deg(p),\deg(q)\}$ where $p/q$ is a rational approximation of $f$ is called the approximate rational degree. It is denoted by $\operatorname{rdeg}_{\epsilon}(f)$. When $\epsilon = 0$, this is exact rational degree $\operatorname{rdeg}(f)$.
\end{definition}

\begin{definition}[Approximate degree]
  A real polynomial $p$ approximates a Boolean function
$f:\{0,1\}^n\to\{0,1\}$ if, for every $x\in\{0,1\}^n$,
$|f(x)-p(x)|\le \epsilon$.
The approximate degree, of $f$, denoted $\deg_\epsilon(f)$, is the
minimum degree of any polynomial approximating $f$. We use $\widetilde{\deg}(f)$ to denote $\deg_{1/3}(f)$.  
\end{definition}

\begin{definition}[Sign degree]
A real polynomial $p$ sign-represents a Boolean function
$f:\{0,1\}^n\to\{0,1\}$ if, for every $x\in\{0,1\}^n$,
$p(x)<0$ iff $f(x)=1$.
The sign degree, or threshold degree, of $f$, denoted $\deg_\pm(f)$, is the
minimum degree of any polynomial sign-representing $f$.
\end{definition}

\begin{definition}[Deterministic Decision Tree Complexity]
    For a Boolean function $f:\{0,1\}^n \rightarrow \{0,1\}$, the deterministic decision tree complexity, denoted as $D(f)$, is defined as the minimum number of queries made to the input $x\in\{0,1\}^n$ by a deterministic algorithm in the worst case to compute $f(x)$.
\end{definition}

\begin{definition}[Certificate Complexity] 
    Let $S \subseteq [n]$ and $C: S \rightarrow \{0,1\}$. $C$ is called an assignment. For an input $x \in \{0,1\}^n$, $C$ is said to be consistent, if $x_i = C(i)$ for $i \in S$. $0-$certificate (or $1-$certificate) is an assignment $C$ such that $f(x)=0$ (or $f(x)=1$) whenever $x$ is consistent with $C$. The certificate complexity $C(f,x)$ is the size of the smallest $f(x)-$ certificate consistent with $x$. $0-$certificate complexity (or $1-$certificate complexity) will be defined as $C^{(0)}(f) = \max_{\{x \in f^{-1}(0)\}}\{C(f,x)\}$ ( or $C^{(1)}(f) = \max_{\{x \in f^{-1}(1)\}}\{C(f,x)\}$). Certificate complexity of $f$ is given by $C(f) =\max\{C^{(0)}(f),C^{(1)}(f)\}$.  
\end{definition}

\begin{definition}[Sensitivity]
   For an input $x \in \{0,1\}^n$ and a Boolean function $f$, the sensitivity of $f$ on $x$ is defined as the number of indices $i \in [n]$, such that, $f(x) \ne f(x^{\{i\}})$. It is denoted as $s(f,x)$. Sensitivity of $f$ on $0-$inputs, denoted as $s_0(f)$, is the maximum sensitivity of $f$ on all $0-$inputs. Similarly, we can define for $1-$sensitivity, denoted as $s_1(f)$. Sensitivity of the function $f$ is defined as $s(f) = \max\{s_0(f),s_1(f)\}$. 
\end{definition}

\begin{definition}[Sensitive Block]
    For a Boolean function $f: \{0,1\}^n \rightarrow \{0,1\}$ and an input $x \in \{0,1\}^n$, a sensitive block is a set $B \subseteq [n]$ is such that $f(x) \ne f(x^B)$.  
\end{definition}

\begin{definition}[Block Sensitivity]
     For a Boolean function $f: \{0,1\}^n \rightarrow \{0,1\}$ and an input $x \in \{0,1\}^n$ the block sensitivity $bs(f,x)$ is the maximum number of pairwise disjoint sensitive blocks $B_i$, that is, $f(x) \ne f(x^{B_i})$. Block sensitivity $bs(f)$ is maximum of $bs(f,x)$ over all $x \in \{0,1\}^n$.
\end{definition}

\begin{definition}[Monotone function]
Let $g:\{0,1\}^n \to \{0,1\}$ be a total Boolean function.
$g$ is monotonically increasing if for all $x,y \in \{0,1\}^n$ such that $x \preceq y$, we have
\[ g(x) \le g(y). \]
$g:\{0,1\}^n \to \{0,1\}$ is monotonically decreasing if for all $x,y \in \{0,1\}^n$ such that $x \preceq y$, we have
\[ g(x) \ge g(y). \]
$g$ is monotone function if it is either monotonically increasing or decreasing.
\end{definition}

\begin{definition}[Unate function]
A Boolean function $f:\{0,1\}^n\to\{0,1\}$ is unate if there exists
$a\in\{0,1\}^n$ and a monotone Boolean function
$g:\{0,1\}^n\to\{0,1\}$ such that
\[ f(x)=g(x\oplus a), \]
where $\oplus$ denotes bitwise XOR.
\end{definition}

\begin{definition}[Monotone Path~\cite{chugh2022decisiontreecomplexityversus}]
    Let $(x^{(1)} , x^{(2)} \dots , x^{(k)})$ be a sequence such that $x^{(i)} \in \{0,1\}^n$ for $i \in [k]$.
    $(x^{(1)} , x^{(2)} \dots , x^{(k)})$ is called monotone path when for every $i \in [k-1]$, there exists some $j\in [n]$ such that 
    \begin{itemize}
        \item $x^{(i)}_j = 0$ and $x^{(i+1)}_j = 1$.
        \item $x^{(i)}_{\ell} = x^{(i+1)}_{\ell}$ for all $\ell \in [n] \backslash \{j\}$.
    \end{itemize}
\end{definition}

\begin{definition}[Alternating number~\cite{chugh2022decisiontreecomplexityversus}]
    Alternating number for a function $f:\{0,1\}^n \rightarrow \{0,1\}$ of a monotone path $P = \{x^{(1)},\dots,x^{(k)}\}$ is defined as 
    \[\left|\{i: i \in \{1,2,\dots,k-1\} \text{ and } f(x^{(i)}) \ne f(x^{(i+1)})\}\right|\]
    $\operatorname{alt}(x,y)$ for $x,y \in \{0,1\}^n$ and $x \preceq y$ is the maximum alternation number of monotone path from $x$ to $y$. 
    $\operatorname{alt}(f)$ is the maximum alternating number of any monotone path from $0^n$ to $1^n$. 
\end{definition}

\begin{definition}[Max term~\cite{lin2016sensitivityconjecturelogrankconjecture}]
    $f : \{0,1\}^n \rightarrow \{0,1\}$. $u \in \{0,1\}^n \backslash \{1^n\}$ is a max term if $x \succ u$ implies $f(x)=f(1^n)$ and $f(x) \ne f(u)$.
\end{definition}

\begin{definition}[Min Term~\cite{lin2016sensitivityconjecturelogrankconjecture}]
    $f : \{0,1\}^n \rightarrow \{0,1\}$. $d \in \{0,1\}^n\backslash \{0^n\}$ is a min term if $x \prec d$ implies $f(x)=f(0^n)$ and $f(x) \ne f(d)$.
\end{definition}

\begin{definition}[Zebra Functions~\cite{chugh2022decisiontreecomplexityversus}]
    A function $f:\{0,1\}^n \rightarrow \{0,1\}$ is called a zebra function if \emph{all} monotone paths from $0^n$ to $1^n$ have same alternation number.
\end{definition}

\begin{definition}[Symmetric Functions]
    A function $f:\{0,1\}^n \rightarrow \{0,1\}$ is called a symmetric Boolean function if for all $x$ and $x'$ such that $|x|=|x'|$, $f(x)=f(x')$.
\end{definition}

\begin{definition}[Read-$k$ DNF]
    Disjunctive Normal Form (DNF) is a formula with OR of ANDs of literals. A read-$k$ DNF is a DNF in which every variable appears at most $k$ times across the terms.
\end{definition}

\begin{definition}[Graph Property]
    Let $\mathcal{G}_n$ be the set of undirected graphs on $n$ vertices. The function $\mathcal{P}:\mathcal{G}_n \rightarrow \{0,1\}$ is called a Graph Property if for $G_1,G_2 \in \mathcal{G}_n$ and $G_1$ and $G_2$ are isomorphic, then $\mathcal{P}(G_1)=\mathcal{P}(G_2)$.
\end{definition}

\begin{definition}[\texorpdfstring{$k$}{k}-Uniform Hypergraph]
For any constant $k$, a $k$-uniform hypergraph $\mathcal{H} = (V,E)$ consists of $V$ vertices and  $E \subseteq 2^{V}$ edges such that for any edge $e \in E$, $|e|=k$. Here $2^{V}$ denotes all subsets of $V$. \end{definition}

\begin{definition}[\texorpdfstring{$k$}{k}-Uniform Hypergraph Property]
    Let $\mathcal{H}_n^{(k)}$ denote the set of all $k$-uniform 
    hypergraphs on vertex set $[n]$, identified with 
    $\{0,1\}^{\binom{n}{k}}$. A function 
    $\mathcal{P}:\mathcal{H}_n^{(k)} \rightarrow \{0,1\}$ is called 
    a $k$-uniform hypergraph property if for any two $k$-uniform 
    hypergraphs $H_1, H_2 \in \mathcal{H}_n^{(k)}$ that are 
    isomorphic (i.e., related by a permutation of $[n]$), we have 
    $\mathcal{P}(H_1) = \mathcal{P}(H_2)$.
\end{definition}
Following Lemma~\ref{nor_lower_bound} gives lower bound on $\mathsf{NOR}$ using Corollary 1 in~\cite{kovacsdeak2026rationaldegreepolynomiallyrelated} which is given below:
\begin{lemma}[Corollary 1 from~\cite{kovacsdeak2026rationaldegreepolynomiallyrelated}]\label{markov_derivative_special_case}
    If $p \in \mathbb{R}[X_1,\dots,X_n]$ and $h>0$ and $p$ has following properties
    \begin{itemize}
        \item $|p(x)|\le h$ for all $x \in \{0,1\}^n$
        \item $p(0^n)=h$
        \item $p(x).p(0^n) \le 0$ for all $x \in \{0,1\}^n$ with $|x|=1$
    \end{itemize}
    Then,
    \[ \deg(p) \ge \sqrt{n/2}. \]
\end{lemma}

\begin{lemma}\label{nor_lower_bound}
    If $f$ is $\mathsf{NOR}: \{0,1\}^n \rightarrow \{0,1\}$, then $\mathsf{N}(\mathsf{NOR}) = \Theta(\sqrt{n})$.
    \begin{proof}
        The upper bound comes from the approximate degree of $\mathsf{NOR}$ which is $O(\sqrt{n})$. We will show the lower bound is $\sqrt{n/2}$. Let $q(x)$ be the approximate non-deterministic polynomial for $\mathsf{NOR}$ with $\deg(q) = \mathsf{N}(\mathsf{NOR})$. Therefore, $-1/3 \le q(x) \le 1/3$ for $x \in \{0,1\}^n \backslash \{0^n\} $ and $|q(0^n)| \ge 1$. Let $p \in \mathbb{R}[X_1,\dots,X_n]$ be such that $p(x)=q(x)^2-1/3$. Then $p(0^n)=h$ for some $h\ge 2/3$ and $-1/3 \le p(x) \le 0$ for all $x \in \{0,1\}^n \backslash \{0^n\}$.  Since, $p(x)$ follows all the properties given in Lemma~\ref{markov_derivative_special_case}, $\deg(p) \ge \sqrt{n/2}$ and thus we get, $\mathsf{N}(\mathsf{NOR})=\deg(q) \ge \frac{\deg(p)}{2}  \ge \sqrt{\frac{n}{8}}$.
    \end{proof} 
\end{lemma}

\section{Main Results}

\subsection{Basic results}

\begin{fact}[Sign degree lower bound]
\label{lem:sign-degree-lower-bound}
For every non-constant Boolean function $f$ and every $0<\epsilon<1$,
\[ M_\epsilon(f)\ge \frac{1}{2}\deg_\pm(f). \]
\end{fact}

\begin{proof}
Let $p$ be an $\epsilon$-approximate non-deterministic polynomial for $f$, and
let $q$ be an $\epsilon$-approximate non-deterministic polynomial for
$\overline f$, with $\deg(p)=\mathsf{N}_\epsilon(f)$ and $\deg(q)=\mathsf{N}_\epsilon(\overline f)$.

Consider the polynomial
\[ r(x):=q(x)^2-p(x)^2. \]
If $f(x)=1$, then $|p(x)|\ge 1$ and $|q(x)|\le\epsilon$, so $r(x)\le \epsilon^2-1<0$. If $f(x)=0$, then $|q(x)|\ge 1$ and $|p(x)|\le\epsilon$, so $r(x)\ge 1-\epsilon^2>0$.
Therefore $r$ sign-represents $f$. Moreover,
\[ \deg(r) \le 2\max\{\deg(p),\deg(q)\} = 2M_\epsilon(f). \]
Hence $\deg_\pm(f)\le 2M_\epsilon(f)$.
\end{proof}

\begin{lemma}[Functions of high sign degree]
\label{lem:high-sign-degree}
Let $\mathcal C$ be a family of Boolean functions such that every
$f:\{0,1\}^n\to\{0,1\}$ in $\mathcal C$ satisfies
$\deg_\pm(f)\ge c n^\alpha$ 
for some constants $c>0$ and $\alpha>0$. Then, for every
$f\in\mathcal C$,
\[ D(f)\le O(M_\epsilon(f)^{1/\alpha}). \]
Thus Conjecture~\ref{approx_ndeg_conjecture} holds for every such
class $\mathcal C$.
\end{lemma}

\begin{proof}
By Fact~\ref{lem:sign-degree-lower-bound},
$M_\epsilon(f) \ge \frac{1}{2}\deg_\pm(f) \ge \frac{c}{2}n^\alpha.$
Therefore
$n \le \left(\frac{2}{c}\right)^{1/\alpha} M_\epsilon(f)^{1/\alpha}.$
Since a deterministic decision tree can always query all $n$ variables, $D(f)\le n \le \left(\frac{2}{c}\right)^{1/\alpha} M_\epsilon(f)^{1/\alpha}$.
\end{proof}

\begin{corollary}[Almost all Boolean functions satisfy Conjecture~\ref{approx_ndeg_conjecture}]
\label{cor:almost-all-functions}
For almost all Boolean functions $f:\{0,1\}^n\to\{0,1\}$, $D(f) \le O(M_\epsilon(f))$.
\end{corollary}

\begin{proof}
It is known that almost all Boolean functions have sign degree at least
$\lfloor n/2\rfloor$~\cite{ANTHONY199591,Saks_1993}. Since $D(f)\le n$, using Lemma~\ref{lem:high-sign-degree} we get the desired bound.
\end{proof}

Next we prove that if a function $g$ is obtained from $f$ by restrictions,
identifications of variables, permutations of variables, and/or negations of input
literals then $M_\epsilon()$ does not increase.

\begin{lemma}[Restrictions and projections]
\label{lem:restrictions-projections}
Let $m\leq n$, and $g:\{0,1\}^m\to\{0,1\}$ be obtained from
$f:\{0,1\}^n\to\{0,1\}$ by substituting each variable $x_i$ by one of
$\{ 0, 1, y_j, 1-y_j\}$ for some $j\in[m]$.  Then
\[ \mathsf{N}_\epsilon(g) \le \mathsf{N}_\epsilon(f) \textit{\quad and thus \quad } M_\epsilon(g)\le M_\epsilon(f). \]
\end{lemma}

\begin{proof}
    Let $p(x_1,\dots,x_n)$ be an $\epsilon$-approximate non-deterministic polynomial
for $f$. Let $\pi:\{0,1\}^m\to\{0,1\}^n$ denote the corresponding substitution
map, so that
$g(y)=f(\pi(y))$.
Define a new polynomial $q(y)=p(\pi(y))$. 
As the substitution cannot increase degree, $\deg(q)\le \deg(p)$. 

Moreover, if $g(y)=0$, then $f(\pi(y))=0$, and hence $|q(y)|=|p(\pi(y))|\le \epsilon$. And if $g(y)=1$, then $f(\pi(y))=1$, and hence $|q(y)|=|p(\pi(y))|\ge 1$. 
Thus $q$ is an $\epsilon$-approximate non-deterministic polynomial for $g$.
This proves
\[ \mathsf{N}_\epsilon(g) \le \mathsf{N}_\epsilon(f). \]
Applying the same argument to $\overline f$ and $\overline g$ gives
\[ \mathsf{N}_\epsilon(\overline g) \le \mathsf{N}_\epsilon(\overline f), \]
and therefore $M_\epsilon(g)\le M_\epsilon(f)$.
\end{proof}

\subsection{Relation with approximate rational degree}
\begin{lemma}
    Let $f : \{0,1\}^n \rightarrow \{0,1\}$ be a Boolean function. Let $\operatorname{rdeg}_{\epsilon}(f)$ be the approximate rational degree of $f$ and $\mathsf{N}_{\epsilon}(f)$ and $\mathsf{N}_{\epsilon}(\overline{f})$ be the approximate non-deterministic degrees of $f$ and $\overline{f}$ respectively. Then,
    \begin{equation}
        \operatorname{rdeg}_{\epsilon}(f) \le 2\max\{\mathsf{N}_{\epsilon}(f),\mathsf{N}_{\epsilon}(\overline{f})\}
    \end{equation}
\end{lemma}
    \begin{proof}
        Let $p(x)$ and $q(x)$ be approximate non-deterministic polynomials of $f$ and $\overline{f}$, respectively, such that $\mathsf{N}_{\epsilon}(f) = \deg(p)$ and $\mathsf{N}_{\epsilon}(\overline{f}) = \deg(q)$. It suffices to prove that $\left|\frac{p(x)^2}{p(x)^2+q(x)^2} -f(x)\right| \le \epsilon$. For case $f(x)=0$, $\left|p(x)^2\right| \le \epsilon$ and $\left|q(x)^2\right| \ge 1$ and therefore, $\left|\frac{p(x)^2}{p(x)^2+q(x)^2}\right| \le \epsilon$. For case $f(x)=1$, $p(x)^2\ge 1$ and $q(x)^2 \le \epsilon$ and so $\frac{p(x)^2}{p(x)^2+q(x)^2}$ is never greater than 1 and if we denote this as the function $h(y_1,y_2) = y_1/(y_1+y_2)$ where $y_1 \ge 1$ and $y_2 \le \epsilon$, this is a strictly increasing function in $y_1$ and decreasing in $y_2$ and $\lim_{y_1 \rightarrow \infty} h(y_1,y_2) = 1$ and $h(y_1,y_2) \in \left[\frac{1}{1+ \epsilon},1\right)$. This means $\frac{p(x)^2}{p(x)^2+q(x)^2}$ is indeed a rational approximation of $f(x)$. 
    \end{proof}

\begin{remark}
    Approximate non-deterministic degree is not polynomially bounded by approximate rational degree. For the Majority function on n bits, $\widetilde{\operatorname{rdeg}}(\mathsf{MAJ}_n) = \Theta( \log{2n})$ as given by~\cite{DBLP:journals/corr/abs-0910-1862}, and $\mathsf{N}(\mathsf{MAJ}_n) =\Omega( \sqrt{n})$.
\end{remark}

\subsection{Monotone Boolean functions}

Conjecture \ref{approx_ndeg_conjecture} is true for monotone Boolean functions. The main result that we need for establishing its proof is that approximate non-deterministic degree of AND being $\Omega(\sqrt{n})$. This can be seen using dual polynomials technique given by~\cite{BT22}. They use the functions to be $\{-1,1\}^n \rightarrow \{-1,1\}$. But we are using $\{0,1\}^n$ and hence will follow~\cite{SHE18} which gives the lower bound on one-sided approximate degree of AND and NOR
on $n$ bits to be $\Omega(\sqrt{n})$.

The following lemma follows same proof technique used to prove polynomial relationship between degree and rational degree for monotone Boolean functions from~\cite{iyer2023rationaldegreebooleanfunctions} but we reproduce the proof for simplicity while using it for our setting.

\begin{lemma}\label{monotone}
    If $f:\{0,1\}^n \rightarrow \{0,1\}$ is a monotone Boolean function, then 
    \[ \max\{\mathsf{N}(f)^2,\mathsf{N}(\overline{f})^2\} \ge \Omega(s(f)).\]
\end{lemma}

    \begin{proof}
        We will assume that $f$ is monotonically increasing without loss of generality.
        We will only prove that $\mathsf{N}(\overline{f})^2 \ge \Omega(s_0(f))$ and thus, $\mathsf{N}(f)^2 \ge \Omega(s_0(\overline{f})) = \Omega(s_1(f))$ will follow. Let $z = s_0(f)$ and $s_0(f,x) = s_0(f)$. This means that if you flip any one of these $z$ sensitive variables of $x$ which are zeros since $f$ is monotonically increasing, then $f$ will become 1. So if all the variables in $x$ are fixed except for the $z$ sensitive zeros, then we get an $OR_z$ function. Now, we know that $\mathsf{N}(\overline{OR_z}) \ge \sqrt{\frac{z}{8}}$ and therefore, $\mathsf{N}(\overline{f}) \ge \mathsf{N}(\overline{OR_z}) \ge \sqrt{\frac{z}{8}} = \sqrt{\frac{s_0(f)}{8}}$. Therefore, we get
        $\Omega(s(f)) \le \max\{\mathsf{N}(f)^2,\mathsf{N}(\overline{f})^2\}$.
    \end{proof}

\begin{corollary}[Monotone functions]\label{lem:monotone}
    If $f:\{0,1\}^n \rightarrow \{0,1\}$ is monotone Boolean function, then \[\widetilde{\deg}(f) \le O\left( \max\{\mathsf{N}(f)^4,\mathsf{N}(\overline{f})^4\}\right) \text{ and }bs(f) = C(f) \le O\left(\max\{\mathsf{N}(f)^2,\mathsf{N}(\overline{f})^2\}\right).\]
\end{corollary}
    \begin{proof}
        For every monotone Boolean functions $f$, $C(f) = s(f) = bs(f)$~\cite{nisan} and for any Boolean function $\widetilde{\deg}(f) \le \deg(f) \le s(f)^2$.
    \end{proof}

\begin{lemma}[Unate functions]
\label{lem:unate-functions}
Let $f:\{0,1\}^n\to\{0,1\}$ be unate. Then
\[ \widetilde{\deg}(f) \le O(M_\epsilon(f)^4). \]
\end{lemma}

\begin{proof}
    Since $f$ is unate, there exists a monotone function $g$ and a vector
$a\in\{0,1\}^n$ such that $f(x)=g(x\oplus a)$.

The map $x\mapsto x\oplus a$ is a bijection obtained by negating input
literals.  So by using Lemma~\ref{lem:restrictions-projections} we know that $M_\epsilon(g) \leq M_\epsilon(f)$  In fact, $M_\epsilon(g) = M_\epsilon(f)$. Also, $\widetilde{\deg}(f)=\widetilde{\deg}(g)$. Thus by using Corollary~\ref{lem:monotone} we obtain the desired result.
\end{proof}

\subsection{Functions with small alternating number}
We follow proof structure from~\cite{lin2016sensitivityconjecturelogrankconjecture} which gave a relation between sensitivity and block sensitivity in terms of alternating number of the function, but we adapt their proof to show an upper bound on certificate complexity in terms of alternating number and approximate non-deterministic degrees of the function and its negation for zebra functions. We will first look at zebra functions and then functions with bounded alternation number $\operatorname{alt}(f)$ where a monotone path $\mathcal{P}$ could have alternating number $0 \le \operatorname{alt}(\mathcal{P}) \le \operatorname{alt}(f)$. For functions of bounded alternation number, we use their technique of bounding block sensitivity of inputs on paths having no alternation.

\subsubsection{Zebra Functions}

\begin{theorem}
\label{zebra}
    Let $f$ be a non-constant zebra function with $\operatorname{alt}(f)=k$ for some $k \ge 1$. Then \[C(f) \le O\!\left(\operatorname{alt}(f)\max\{\mathsf{N}(f)^2,\mathsf{N}(\overline{f})^2\}\right)\] where $\mathsf{N}$ is the approximate non-deterministic degree.
\end{theorem}
\begin{proof}
    Let us look at any $x \in \{0,1\}^n$. It will satisfy at least one of the following conditions:
    \begin{enumerate}
        \item $x \preceq u$ for some max term $u$.
        \item $x \succeq d$ for some min term $d$.
    \end{enumerate}
    Let us first look at the case where alternation number is odd which means $f(0^n) \ne f(1^n)$.
    Since $f$ is non-constant, there exists $x \ne 1^n$ such that $f(x)\ne f(1^n)$. For such $x$, there has to be a maximal $u \in \{0,1\}^n$ such that $x \preceq u$ and $f(u) \ne f(1^n)$ and this $u$ is then one of the max terms of $f$.
    Now, for the case when $f(x)=f(1^n)$, again since $f$ is non-constant, there exists a $x' \in \{0,1\}^n$ such that $f(0^n) \ne f(x')$ and $x' \preceq x$. Minimal of such $x'$ will be a min term $d \in \{0,1\}^n$. Therefore, in this case, $x$ follows at least one of the above two conditions.

    Now, let us look at the case when the alternation number is even and this means $f(0^n)=f(1^n)$. Let us first look at $x$ such that $f(x) \ne f(1^n)$. For such $x$ there has to both a maximal $u \in \{0,1\}^n$ with $x \preceq u$ and minimal $d \in \{0,1\}^n$ with $d \preceq x$ such that $u$ and $d$ are the max term and min term respectively. Now, for $x$ such that $f(x) = f(0^n) = f(1^n)$, since every monotone path has same alternation number and $\operatorname{alt}(f) > 1$ in this case, either $\operatorname{alt}(0^n,x) > 1$ or $\operatorname{alt}(x,1^n) > 1$. In case of $\operatorname{alt}(0^n,x) > 1$, there has to be a min term $d \preceq x$ and in case of $\operatorname{alt}(x,1^n) > 1$, there has to be a max term $x \preceq u$. So again, $x$ has to follow one of the two above conditions.

    For the rest of the proof, we will be using induction on $\operatorname{alt}(f)$. We know that for $\operatorname{alt}(f)=1$, $f$ is monotone and therefore, from Lemma~\ref{monotone} we have the result. For induction hypothesis, let us assume the result is true for every $\operatorname{alt}(f)<k$. We will now prove $C(f,x) \le O\left( \operatorname{alt}(f) \max\{\mathsf{N}(f)^2,\mathsf{N}(\overline{f})^2\}\right)$ for $\operatorname{alt}(f)=k$.
    
    Let us first look at the case when $x \preceq u$ for some max term $u$. We know that for every $z \succ u$, $f(z)\ne f(u)$. That means every $0$ in $u$ is sensitive. Let $S_0(u) =\{i \in [n] \mid u_i=0\}$. Therefore,

    \[ |S_0(u)| \le s(f,u). \]
    Let us now look at restriction of $f$ to the sub-cube $\{z \succeq u\}$ by keeping the variables at $S_0(u)$ indices free. This restriction $f'$ is clearly obtained by fixing the $1's$ in $u$. Note that $f'$ is a monotone function. Therefore, $s(f') \le O(\max\{\mathsf{N}(f')^2,\mathsf{N}(\overline{f'})^2\})$. Also, $s(f',0^{|S_0(u)|}) = |S_0(u)| \le s(f') \le O(\max\{\mathsf{N}(f')^2,\mathsf{N}(\overline{f'})^2\})$. So we get an equation which will be used later.
    \begin{equation}\label{zero_indices}
        |S_0(u)| \le O(\max\{\mathsf{N}(f')^2,\mathsf{N}(\overline{f'})^2\}) \le O(\max\{\mathsf{N}(f)^2,\mathsf{N}(\overline{f})^2\}).
    \end{equation}

    We will now focus on the restriction of $f$ on the sub-cube $\{z \preceq u\}$. This restriction $f''$ is obtained by fixing the $S_0(u)$ indices to be zero and leaving the rest of the variables free. For every such $x$ in this sub-cube, we will have,
    \begin{equation}\label{certificate_complexity_restriction_sensitivity}
        C(f,x) \le C(f'',x^{[n]-S_0(u)}) + |S_0(u)|.
    \end{equation}
    where $x^{[n]-S_0(u)}$ denotes the bits of $x$ at all indices $[n]$ but $S_0(u)$.
    This is because let $C(f'',x^{[n]-S_0(u)})$ be the smallest size of an $f''(x^{[n]-S_0(u)})$-certificate consistent with $x^{[n]-S_0(u)}$ and so it is an assignment $C'' : S \subseteq \{[n]-S_0(u)\} \rightarrow \{0,1\}$. Since $f''$ is a restriction of $f$ over subcube $\{x \preceq u\}$, this means that $f''(x^{[n]-S_0(u)})=f(x)$ for $x \preceq u$ and the coordinates $x_i=0$ when $i\in S_0(u)$. Now, if we take an assignment $C'$ such that
    \begin{equation*}
        C'_i = \begin{cases}
        C''_i  & \text{for } i \in {[n]-S_0(u)}, \\
        0 & \text{for }i \in S_0(u).
        \end{cases}
    \end{equation*}
    
    Then $C'$ will be an $f(x)$-certificate, consistent with $x$. 
    We also know that $\operatorname{alt}(f'') =  \operatorname{alt}(f)-1$.  
    This is because every path from $u$ to $1^n$ has to have alternation $1$ since $f(u) \ne f(1^n)$. Therefore, the path from $0^n$ to $u$ will have alternation $\operatorname{alt}(f)-1$ to make sure the alternation from $0^n$ to $1^n$ is $\operatorname{alt}(f)$.
    Also note that $f''$ is a zebra function. Let us assume that it was not a zebra function. Then there would exist paths $P$ and $P'$ from $0^n$ to $u$ with $\operatorname{alt}(P) \ne \operatorname{alt}(P')$. But since we know that the path from $u$ to $1^n$ has alternation 1, there would exist two paths from $0^n$ to $1^n$ with alternations $\operatorname{alt}(P)+1$ and $\operatorname{alt}(P')+1$ which would contradict the fact that $f$ is a zebra function. Hence $f''$ has to be a zebra function.
    So, we can use the induction hypothesis on $f''$ to get $C(f'') \le O(\operatorname{alt}(f'')(\max\{\mathsf{N}(f'')^2,\mathsf{N}(\overline{f''})^2\}))$.
    Also, since $f''$ is a restriction of $f$, $\max\{\mathsf{N}(f'')^2,\mathsf{N}(\overline{f''})^2\}) \le \max\{\mathsf{N}(f)^2,\mathsf{N}(\overline{f})^2\})$ and we get,
    \begin{equation} \label{lesser_subcube_induction}
        C(f'',x^{[n]-S_0(u)}) \le  O((\operatorname{alt}(f)-1)(\max\{\mathsf{N}(f)^2,\mathsf{N}(\overline{f})^2\})).
    \end{equation}
    Putting equations \eqref{zero_indices} and \eqref{lesser_subcube_induction} into equation \eqref{certificate_complexity_restriction_sensitivity}, we get for $x \preceq u$,
    \begin{equation}
        C(f,x) \le O(\operatorname{alt}(f) (\max\{\mathsf{N}(f)^2,\mathsf{N}(\overline{f})^2\})).
    \end{equation}

To prove the case when $x \succeq d$, we will look at the function $g$ such that $f(x) = g(\overline{x})$. We clearly have $C(g,\overline{x}) = C(f,x)$, $\operatorname{alt}(f)=\operatorname{alt}(g)$ and $(\max\{\mathsf{N}(f)^2,\mathsf{N}(\overline{f})^2\}) = (\max\{\mathsf{N}(g)^2,\mathsf{N}(\overline{g})^2\})$. $d \preceq x$ means $d_i=1$ implies $x_i=1$. This translates to $\overline{d_i}=0$ implies $\overline{x_i}=0$. Therefore, $\overline{d} \succeq \overline{x}$. This proves the existence of a max term $u' \succeq \overline{d}$ for $\overline{x}$ for function $g$. Therefore, we can use the previous case to get $C(g,\overline{x}) \le O(\operatorname{alt}(g) \max\{\mathsf{N}(g)^2,\mathsf{N}(\overline{g})^2\})$. Hence, when $d \preceq x$, finally we get
\[C(f,x) \le O(\operatorname{alt}(f)(\max\{\mathsf{N}(f)^2,\mathsf{N}(\overline{f})^2\})).\]
Both these cases covers all $x \in \{0,1\}^n$ and hence we have 
\[ C(f) \le O(\operatorname{alt}(f)\max\{\mathsf{N}(f)^2,\mathsf{N}(\overline{f})^2\}).\]
\end{proof}

\begin{corollary}\label{zebra_D_approxndeg}
    If $f$ is a non-constant zebra Boolean function with $\operatorname{alt}(f)=k$, then \[D(f) \le O(k^2 \max\{\mathsf{N}(f)^4,\mathsf{N}(\overline{f})^4\}).\]
\end{corollary}

    \begin{proof}
        We know that $D(f) \le C^{(0)}(f)C^{(1)}(f)$~\cite{survey_complexity_measures}. Using the previous Theorem~\ref{zebra} and using the definition $C^{(0)}(f) \le C(f)$ and $C^{(1)}(f) \le C(f)$, we have 
        \begin{equation}
            D(f) \le O(k^2 \max\{\mathsf{N}(f)^4,\mathsf{N}(\overline{f})^4\}).
        \end{equation}
    \end{proof}
    
The corollary below states that the Conjecture~\ref{approx_ndeg_conjecture} is true for total zebra functions having small alternating number.
\begin{corollary}
    Let $f: \{0,1\}^n \rightarrow \{0,1\}$ have $\operatorname{alt}(f) \le O(poly \log{n})$. Then, $D(f) \le \widetilde{O} (\mathsf{N}(f)^4\mathsf{N}(\overline{f})^4)$
\end{corollary}
    \begin{proof}
        The proof is straight-forward from previous Corollary~\ref{zebra_D_approxndeg}.
    \end{proof}

\subsubsection{Functions with bounded alternating number}

We will now look at functions which are not necessarily zebra but have bounded alternation number $\operatorname{alt}(f)=k$. This proof is similar to~\cite{lin2016sensitivityconjecturelogrankconjecture} where they use the bound on zero block and one block parts of the sensitive blocks of input $x$ which is on monotone path with no alternations. We repeat the proof by using approximate non deterministic degree instead of sensitivity and the proof is given in the Appendix~\ref{alternating_number_appendix}. 
\begin{theorem}\label{alternating_number_proof_2}
    Let $f:\{0,1\}^n \rightarrow \{0,1\}$ be alternation number $k$. Then, $bs(f) \le O\left(k \max\{\mathsf{N}(f)^2,\mathsf{N}(\overline{f})^2\}\right)$.
\end{theorem}

\begin{corollary}\label{cor:bounded-alt-D}
For every Boolean function $f$,
\[ D(f) \le O(\operatorname{alt}(f)^3\max\{\mathsf{N}(f)^6,\mathsf{N}(\overline{f})^6\}). \]
\end{corollary}

\begin{proof}
By Theorem~\ref{alternating_number_proof_2} we have  $bs(f)\le O(\operatorname{alt}(f)\max\{\mathsf{N}(f)^2,\mathsf{N}(\overline{f})^2\})$. Using the standard bound $D(f)\le O(bs(f)^3)$, the result follows.
\end{proof}

\subsection{Symmetric functions}

\begin{lemma}\label{symm_alt}
    Let $f: \{0,1\}^n \rightarrow \{0,1\}$ be an arbitrary symmetric Boolean function and let $\operatorname{alt}(f)$ be its alternating number. Then, 
    \[ \mathsf{N}(f) \ge \frac{\operatorname{alt}(f)}{2}.\]
\end{lemma}
    \begin{proof}
        Let $p(x)$ be approximate non-deterministic polynomial for $f(x)$ with $ \mathsf{N}(f) = \deg(p)=d$ and polynomial $p': \mathbb{R} \rightarrow \mathbb{R}$ be such that $p'(k) = \mathbb{E}_{|x|=k}\left[p(x)^2\right]$ and so $\deg(p') \le 2d$. Therefore, when $f(x)=0$, $0 \le p'(|x|) \le 1/3$ and $f(x)=1$ implies $p'(|x|) \ge 1$. $\operatorname{alt}(f)$ in case of symmetric Boolean functions will be 
        \[ \operatorname{alt}(f) = |K_1 \cup K_2|\]
        where \[K_1 = \left\{k \in \{0,\dots,n-1\} \mid p'(k) \le 1/3 \, , p'(k+1) \ge 1\right\} \text{ and }\] \[ K_2 = \{k \in \{0,\dots,n-1\} \mid p'(k) \ge 1 \, ,  p'(k+1) \le 1/3\}.\]
        Note that  $p''(k) = p'(k)-2/3$ have number of roots at least equal to or more than $\operatorname{alt}(f)$. Therefore, $\mathsf{N}(f) = \deg(p) \ge \frac{\deg(p')}{2}= \frac{\deg(p'')}{2} \ge \frac{\operatorname{alt}(f)}{2}$.  
    \end{proof}

\begin{corollary}\label{symmetric_relation}
     Let $f: \{0,1\}^n \rightarrow \{0,1\}$ be an arbitrary symmetric Boolean function. Then, 
    \[ D(f) \le O(\max\{\mathsf{N}(f)^6,\mathsf{N}(\overline{f})^6\}) .\]
\end{corollary}
    \begin{proof}
        Let $\operatorname{alt}(f)$ be the alternating number of $f$. From the earlier Lemma~\ref{symm_alt}, we have $\max\{\mathsf{N}(f),\mathsf{N}(\overline{f})\} \ge \frac{\operatorname{alt}(f)}{2}=\frac{\operatorname{alt}(\overline{f})}{2}$. Symmetric functions are zebra functions and hence we can use 
        \[D(f) \le O(\operatorname{alt}(f)^2 \max\{\mathsf{N}(f)^4,\mathsf{N}(\overline{f})^4\})\] from Corollary \ref{zebra_D_approxndeg}.
        Thus, 
        \[D(f) \le O(\max\{\mathsf{N}(f)^6,\mathsf{N}(\overline{f})^6\}).\]
    \end{proof}

\subsubsection{Improved bounds for symmetric functions}

In this section, we look at those symmetric functions where the function is constant for a large interval of Hamming weights around $n/2$ Hamming weight. We refer to the technique of reduction to the $EXACT_0$ function used in~\cite{10.1109/CCC.2008.18}.
Let $EXACT_k(t)$ be a univariate function on $t \in \{0,\dots,n\}$ such that
\begin{equation*}
        EXACT_k(t) = \begin{cases}
        1  & \text{for } t=k \\
        0 & \text{otherwise}       
        \end{cases}
\end{equation*}

\begin{lemma}\label{improved_symm_lemma}
    Let $f : \{0,1\}^n \rightarrow \{0,1\}$ be a symmetric Boolean function. Let $D: \{0,1,\dots,n \} \rightarrow \{0,1\}$ be such that $f(x) = D(|x|)$. 
Let $\ell_0(D) \in \{0,1,\dots,\lfloor\frac{n}{2}\rfloor\}$ and $\ell_1(D) \in \{0,1,\dots,\lfloor \frac{n}{2}\rfloor\}$ be the minimum integers such that $D$ is constant in $[\ell_0(D),n-\ell_1(D)]$. When at least one of $\ell_0$ or $\ell_1$ is less than $n/5$, then $\max\{\mathsf{N}(f),\mathsf{N}(\overline{f})\} = \Omega(\sqrt{n})$.
\end{lemma}
\begin{proof}
Let us first consider $\ell_0(D) < \frac{n}{5}$. The proof for the case when $\ell_1 
< \frac{n}{5}$ will be analogous. Let $\ell = \ell_0(D)$ and we assume $\ell \ge 1$ since if $\ell=0$, then we can choose $x$ such that $|x| = n-\ell_1(D
)+1$ and embed an AND or NAND, depending on value of $D$, on the ones of $x$ to get desired result. 
According to~\cite{10.1109/CCC.2008.18}, $EXACT_0 (t)$ on $\lfloor \frac{n}{5} \rfloor$ bits can be written as 
\[EXACT_0(t) = (1-2D(\ell))D(t+\ell-1)+D(\ell).\] 

Let $q(x)$ be an approximate non-deterministic polynomial for $f(x)$ with $\deg(q)=\mathsf{N}(f)$ and $p(t)$ be the symmetrized version of $q(x)^2$ where $p(t) = \mathbb{E}_{|x|=t}[q(x)^2]$. So, when $D(t) =1$, then $p(t) \ge 1$ and when $D(t)=0$ then $0 \le p(t) \le 1/3 $. If we show that $(1-2D(\ell))p(t+\ell-1)+D(\ell)$ is the approximate non-deterministic polynomial for $EXACT_0(t)$ and that the approximate non-deterministic degree for $EXACT_0(t)$ is lower bounded by $\Omega(\sqrt{n})$, then we will get $\mathsf N(f) = \Omega(\sqrt{n})$. $D(t)$ will satisfy exactly one of the following two cases:

\begin{enumerate}
    \item $D(\ell) =0$ and $D(\ell -1) =1$.
    \item $D(\ell) =1$ and $D(\ell -1) =0$.
\end{enumerate}

It suffices to show $\mathsf N(f) = \Omega(\sqrt{n})$ for the first case. This is because when $f$ is such that $D(\ell) =1$ and $D(\ell -1) =0$, then $\overline{D}(\ell)=0$ and $\overline{D}(\ell -1)=1$ and hence $\mathsf N(\overline{f}) = \Omega(\sqrt{n})$.

Let us now look at the case when $D(\ell)=0$ and $D(\ell-1)=1$.
We know $EXACT_0(0)=1$. Let $p(\ell -1)=h$. Since $D(\ell-1)=1$, $h\ge 1$. The value of the expression $(1-2D(\ell))p(t+\ell-1)+D(\ell)$ when $t=0$ becomes 
\[(1-2D(\ell))p(\ell-1)+D(\ell) = (1-0)h+0 = h \ge 1.\]

For $1\le t \le \lfloor n/5 \rfloor$, 
$p(\ell+t-1) = \epsilon \le 1/3$. Then the value of the required expression 
\[(1-2D(\ell))p(t+\ell-1)+D(\ell) = \epsilon.\] Also we know that $EXACT_0(t) = 0$ for $t \ge 1$. We can conclude that $(1-2D(\ell))p(t+\ell-1)+D(\ell)$ can be an approximate non-deterministic polynomial for $EXACT_0(t)$ and its degree is at most $2N(f)$. Therefore, $N(f) = \Omega(N(EXACT_0))$.

Now, let us prove the lower bound on the approximate non-deterministic degree of $EXACT_0(t)$ is $\Omega(\sqrt{n})$. Let $p''(t)$ be an approximate non-deterministic polynomial for $EXACT_0(t)$ with $\deg(p'') = N(EXACT_0)$. So we have $p''(0) =h \ge 1$ and $|p''(t)| \le 1/3$ for $t=1,\dots,\lfloor n/5 \rfloor$. Now, because for this specific function $EXACT_0(t)$, the value of the polynomial $p''(t)$ is greater than 1 only at $t=0$, we can just scale the polynomial by the value $h$ without losing the property of it being the approximate non-deterministic polynomial. So, lets say $p'(t) = p''(t)/h$, then $\deg(p'') = \deg(p')$. But now $p'$ is actually the symmetrized polynomial for the two-sided approximate degree of NOR. From the proof of lower bound of approximate degree of symmetric functions from~\cite{10.1145/129712.129758}, we already know that $ \widetilde{\deg}(\mathsf{NOR}) \ge \deg(p') \ge \Omega(\sqrt{n})$.   

So, to conclude, if $D$ is such that it follows the first case, then we have $\mathsf{N}(f) \ge \Omega( \sqrt{n}
)$. Suppose $D$ is such that it follows the second case, in that case $\overline{D}$, which is the predicate corresponding to $\overline{f}$ will follow the first case and so $\mathsf{N}(\overline{f}) = \Omega(\sqrt{n})$. Hence, $\max\{\mathsf{N}(f),\mathsf{N}(\overline{f})\} \ge \Omega(\sqrt{n})$.
\end{proof}

\begin{corollary}\label{cor:symmetric-stronger}
Let $f : \{0,1\}^n \rightarrow \{0,1\}$ be a symmetric Boolean function. Let $D: \{0,1,\dots,n \} \rightarrow \{0,1\}$ be the function such that $f(x) = D(|x|)$. 
Let $\ell_0(D)$ and $\ell_1(D)$ be as described before and when at least one of $\ell_0$ or $\ell_1$ is less than $n/5$, then $\widetilde{\deg}(f) \le O(\max\{\mathsf{N}(f)^2,\mathsf{N}(\overline{f})^2\}) $.    
\end{corollary}
\begin{proof}
    Since $\widetilde{\deg}(f) \le \deg(f) \le n$, using previous Lemma~\ref{improved_symm_lemma}, we get our result.
\end{proof}

\subsection{Read-once functions}
In this section, we see that the conjecture is true for read-once functions which are functions where every variable can only appear at most once.
Following the same proof techniques of~\cite{iyer2023rationaldegreebooleanfunctions}, we get the following two lemmas:

\begin{lemma}\label{read_once_composition}
        Let $f: \{0,1\}^n \rightarrow \{0,1\}$ and $g_i : \{0,1\}^{n_i} \rightarrow \{0,1\}$ be Boolean functions such that every variable is relevant. Then $h:\{0,1\}^{\sum_i n_i} \rightarrow \{0,1\}$ be such that $h(x^1,x^2,\dots,x^n) = f(g_1(x^1),\dots,g_n(x^n))$. Then,
    \[\max\{\mathsf{N}(h),\mathsf{N}(\overline{h})\} \ge \max\{\mathsf{N}(f),\mathsf{N}(\overline{f}),\mathsf{N}(g_1),\mathsf{N}(\overline{g_1}),\dots, \mathsf{N}(g_n),\mathsf{N}(\overline{g_n})\}.\]
\end{lemma}
\begin{proof}
    Let us first look at the function $g_i$. Since $g_i$ is relevant in each variable, there exists a restriction $g_i'$  where we can fix every variable of $g_i$ except for one variable $x^j_i$. Then, $g_i$ could either be $x^i_{j_i}$ or $\overline{x^i_{j_i}}$. We can find restrictions for each $g_i$ and since the variables are disjoint, by taking the combined restrictions, we get a restriction $h'$. We can now write $h'$ as $h' = f(x^1_{j_1},\dots,x^n_{j_n})$. Therefore, we will get $M(h) \ge M(h')$ and hence $M(h) \ge M(f)$. For showing lower bound with respect to $g_i$, we can set all $x^i$ except $x^k$ such that $f(x)=g(x^k)$ or $f(x)=\overline{g}(x^k)$. From this, we get a restriction $h''$ for $h$ such that $h''=g_k(x^k)$ or $h''=1-g_k(x^k)$. Therefore, $M(h) \ge M(g_k)$. Since this $k$ was chosen arbitrarily, $M(h) \ge \max\{M(g_1),\dots,M(g_n)\}$. 
    Along with previous lower bound with respect to $f$, we finally get \[M(h) \ge \max\{M(f),M(g_1),\dots,M(g_n)\}.\] 
\end{proof}

\begin{lemma}\label{read_once}
    Let $f: \{0,1\}^n \rightarrow \{0,1\}$ be a read-once formula with symmetric gates and depth $d$. Then $\max\{\mathsf{N}(f)
    ,\mathsf{N}(\overline{f})\} \ge \Omega(\deg(f)^{1/6d})$.
\end{lemma}
\begin{proof}
        Let $G$ be the node that has maximum branching factor $w$ and $F$ be the sub formula of $f$ with $G$ as root node. Let $F_1,\dots,F_w$ be the sub-formulae corresponding to each of the nodes under $F$. There exists a restriction for each $F_i$ where we can fix every variable of $F_i$ except for one variable $x^i_{j_i}$ under $F_i$. Because each $F_i$ are read-once and $F$ is read once, these restrictions are on disjoint variables and we can combine them to get a restriction $F|_r = G(x^1_{j_1},\dots,x^w_{j_w})$. Using similar technique from previous lemma ~\ref{read_once_composition}, an additional restriction can be used to obtain $f|_{r'} = F$ and hence $f|_{r' \cup r} = F|_r$. Again, since we are working with read-once formulas, restriction corresponding to $r'$ will not affect the restriction corresponding to $r$. From Lemma~\ref{lem:restrictions-projections}, we get $\max{\{\mathsf{N}(f),\mathsf{N}(\overline{f})\}} \ge \max\{\mathsf{N}(f|_{r \cup r'}),\mathsf{N}(\overline{f}|_{r \cup r'})\} = \max\{\mathsf{N}(F|_{r}),\mathsf{N}(\overline{F}|_{r})\}=\max\{\mathsf{N}(G),\mathsf{N}(\overline{G})\}$. 
        Since $G$ is symmetric, using Corollary~\ref{symmetric_relation}, we get $\max\{\mathsf{N}(G),\mathsf{N}(\overline{G})\} \ge \Omega(D(G)^{1/6})$ and using lower bounds on symmetric functions $D(G) \ge \Omega(w)$~\cite{survey_complexity_measures}, we have $\max\{ \mathsf{N}(G), \mathsf{N}(\overline{G})\} \ge \Omega(D(G)^{1/6}) \ge \Omega(w^{1/6})$. We know $\deg(f) \le n$ and $w > n^{1/d}$. This is because if the maximum branching factor is less than $n^{1/d}$, then the total number of literals will be less than $n$. Hence, we finally get $\max\{\mathsf{N}(f)
    ,\mathsf{N}(\overline{f})\} \ge \Omega(\deg(f)^{1/6d})$.
    \end{proof}

A similar result was shown for one-sided approximate degree of AND-OR functions by~\cite{SHE18}.

\subsection{Read-\texorpdfstring{$k$}{k} DNF formulas}

 We adapt the proof for non-deterministic degree from~\cite{dewolf2004nondeterministicquantumqueryquantum} for approximate non-deterministic degree lower bound in terms of minimal sensitive block size.  
\begin{lemma}\label{minimal_sensitive_lower_bound}
    Let $f: \{0,1\}^n \rightarrow \{0,1\}$ and $x \in \{0,1\}^n$ be such that $f(x)=0$. Let $B$ be the minimal sensitive block of $x$. Then $ \mathsf{N}(f) \ge \Omega(\sqrt{|B|})$.
\end{lemma}
    \begin{proof}
        We can assume for $x=0^n$, $f(x)=0$. Let $B$ be minimal sensitive block of $x$ and therefore $f(x^B)=1$ but for every $B' \subset B$, $f(x^{B'})=0$. Fixing all the variables outside of $B$ to zero, we get $f'$ as AND function on $|B|$ variables, and therefore, $\mathsf{N}(f) \ge \mathsf{N}(f') \ge \Omega(\sqrt{|B|})$.
    \end{proof}

For Lemma~\ref{lem:ndeg-alpha}, we use the number of disjoint terms argument from~\cite{bafna2016sensitivity} and then use Lemma~\ref{minimal_sensitive_lower_bound} to give lower bound on approximate non-deterministic degree of $\overline{f}$.

\begin{lemma}\label{lem:ndeg-alpha}
    Let $f: \{0,1\}^n \rightarrow \{0,1\}$ be a Boolean function computed by a minimal read-$k$ DNF formula with size $\alpha$ and maximum width $\beta$. Then, 
    \[ \mathsf{N}(\overline{f}) \ge \Omega\Bigl(\sqrt{\frac{\alpha}{k \beta}}\Bigr).\]
\end{lemma}
\begin{proof}
        Let $f$ be computed by a $read-k$ $DNF$ with maximum width $\beta$ and size $\alpha$. This means that the OR will have a fan-in of $\alpha$ and each bottom AND will have a fan-in of $\beta_i$ with $\max_i\{\beta_i\}=\beta$. We can now select the terms that have disjoint variables. We will be left with $\frac{\alpha}{k \beta}$ terms. Now, since these have disjoint variables, when each of the positive literals in these terms take value $1$ and each of the negative literals take value $0$, each of these terms evaluates to $1$ and the function value is $1$ for such input lets call it $x'$. It should not matter what other terms evaluate to, $f(x') = 1$ and therefore, $\overline{f}(x')=0$. Minimum sensitive block for $x'$ will have at least one literal from each of these terms. Therefore, minimum sensitive block size for $x'$ will be $\Omega\Bigl(\frac{\alpha}{k \beta}\Bigr)$.
        Using Lemma~\ref{minimal_sensitive_lower_bound}, we obtain, $\mathsf{N}(\overline{f}) = \Omega\Bigl(\sqrt{\frac{\alpha}{k \beta}}\Bigr)$.
    \end{proof}

\begin{lemma}\label{lem:read-k}
Let $f$ be computed by a minimal read-$k$ DNF. Then
\[
\mathsf{N}(\overline f)
\ge \Omega\left(\sqrt\frac{s_0(f)}{2k+1}\right) \textit{\quad and \quad }
\mathsf{N}(f)
\ge \Omega\left(\sqrt\frac{s_1(f)}{k+1}\right).
\]
\end{lemma}
\begin{proof}
    Let $x\in f^{-1}(0)$ be the input that satisfies $s_0(f,x)=s_0(f)=t$. Let $S \subseteq [n]$ be the set of sensitive coordinates of $x$. Thus $|S|=t$, and for every $i\in S$,
$f(x^{\{i\}})=1$.
For each $i\in S$, choose a DNF term $T_i$ that is satisfied by
$x^{\{i\}}$. As $i$ is sensitive, such $T_i$ exists.  Moreover, the term $T_i$ is not satisfied by $x$ as $f(x)=0$. 
Since $T_i$ is not satisfied by $x$, the term $T_i$
contains the variable $i$ with the literal opposite to its value in $x$. Also, every other literal of $T_i$ is satisfied by $x$.

Now consider a graph on these $S$ vertices.  For every $i \neq j$, put the edge $(i,j)$ if $T_i$ contains variable $j$ or $T_j$ contains variable $i$. But because the DNF formula is read-$k$ each variable $i$ can at most be in $k$ many terms and there are $t$ variables in $S$.
Hence there are at most $kt$ incidences of variables from $S$ inside the terms $\{T_i : i \in S \}$. 
As each edge comes from at least one such incidence, there are at most $kt$ edges in total in the graph. 
Hence the graph has an independent set $I$ of size at least $\frac{t}{(2k+1)}$. 
This means that for any distinct $i, j \in I$, term $T_i$ does not contain variable $j$. 

Now set all variables to their values in $x$. As $f(x)=0$, the DNF formula evaluates to 0. 
We now consider flipping subsets of the $I$ variables and argue that the DNF will always evaluate to 1: For any nonempty subset
$A \subseteq I$, consider any $j \in A$. Since $I$ is an independent set, the term $T_j$
contains no other variable $\ell \neq j$ from $I$. The literal of $T_j$ on variable $j$ is
opposite to $x_j$, so flipping $j$ satisfies it, and all other literals of
$T_j$ are already satisfied by $x$. Hence $T_j$ evaluates to $1$, and
therefore $f = 1$. The only assignment to the $I$-coordinates that keeps
$f = 0$ is the original $x$-assignment. Thus the restricted function on $I$
is $\mathrm{OR}_{|I|}$, up to negations of input variables.

Hence $\overline{f}$ restricted to $I$ is $\mathrm{AND}_{|I|}$, up to
negations. Since restrictions do not increase approximate non-deterministic
degree (Lemma~\ref{lem:restrictions-projections}), and $\mathsf{N}(\mathrm{AND}_m) = \Omega(\sqrt{m})$, we obtain
\[
\mathsf{N}(\overline{f})
= \Omega(\sqrt{|I|})
= \Omega\!\left(\sqrt{\frac{s_0(f)}{2k+1}}\right).
\]

For the second inequality, let $x\in f^{-1}(1)$ satisfy $s_1(f,x)=s_1(f)=t$, and let $S \subseteq [n]$ be the set of sensitive coordinates of $x$. Thus $f(x)=1,$ and $f(x^{\{i\}})=0$ for every $i \in S$. Choose a DNF term $T^*$ that is satisfied by the $x$. 
Since flipping any $i \in S$ changes the value of the function from 1 to 0, the term $T^*$ must contain every variable $i \in S$. Otherwise $T^*$ would remain satisfied by $x^{\{i\}}$. 

Now set all variable according to $x$. Our goal is to embed AND in some of the variables in $S$. The problem is, if we flip several variables in $S$, the term $T^*$ becomes false, but some other term might become true. We remove this by again considering an independent set, but in a slightly different way. 

For any term $T' \neq T^*$, define a set of potentially problematic  variables:

\[ \mathcal{B}(T') = \{ i \in S: T' \textit{ contains  variable } i \textit{ with literals opposite to its value in } x \}.\]

For example, for any $i \in S$, $i \in \mathcal{B}(T')$ if $ x_i=1$ and $\overline{x_i} \in T'$. 

If $T'$ becomes true after flipping some variables in $S$, then $\mathcal{B}(T') \neq \emptyset$. However, as for all $i \in S$, $f(x^{\{i\}}) =0$, $|\mathcal{B}(T')|$ cannot be 1 and we need to flip more than one variables in $S$ to get such potential problematic cases. Thus $|\mathcal{B}(T')|\geq 2$. 

Now let $\mathcal{T}'$ be the set of all such $T'$ with $|\mathcal{B}(T')|\geq 2$. We call these set of terms potentially problematic terms. 
For each $T'\in \mathcal{T}'$, choose exactly one pair of variables from $\mathcal{B}(T')$, and put a single edge between them. This gives a graph $G$ on vertex set $S$.

Since the DNF is read-$k$, each variable appears in at most $k$ terms. Moreover, each term contributes at most one chosen edge. Therefore every vertex
of $G$ is incident to at most $k$ edges, and hence
the maximum degree of $G$, $\Delta(G)\le k$.
Thus $G$ has an independent set $I\subseteq S$ of size
\[
|I|\ge \frac{|S|}{k+1}
=
\frac{s_1(f)}{k+1}.
\]

Now if we fix all variables in $[n]\setminus I$ according to the input $x$, the resulting function on the variables in $I$ is
$\mathrm{AND}_{|I|}$, up to negations of input variables. To see this notice that, $T^*$ contains every variable $i\in I$ with the literal satisfied by $x_i$. Therefore $T^*$ is satisfied exactly when all variables in $I$ remain equal to their values in $x$.

Crucially, no $T' \ne T^*$ can become true after flipping any nonempty subset of variables from $I$. We are only concerned with the terms that are false for $x$ but can become true by flipping subsets of $I$. Suppose the term $T''$ is false because of variables outside of $I$. Then $T''$ will continue to remain false even after we flip variables in any subset of $I$. For the remaining term $T' \ne T^*$, if $\mathcal{B}(T') = \emptyset$, then there are two possibilities for $T'$, either $T'$ does not contain any $x_i$ with $i \in S$ or if it contains variables from $S$, then those are same literals as those in $x$. In the first case where it does not contain any $x_i$ with $i \in S$, flipping subsets of $I$ will not affect the term, so again the term remains to be false. Let us now look at the second case where if it contains variables from $S$, then those are same literals as $x$. We know the value of this term $T'$ is originally false for $x$ and hence there exists a literal $x_i \in T'$ such that $x_i=0$. It could also be the case that $\overline{x_i}$ is contained in the term $T'$ in which case $x_i=1$. Because $T^*$ contains all variables corresponding to $S$, all variables in $S$ contribute to term $T^*$ being $1$, but since the term $T'$ is false, that means $i \notin S$ and this variable will not be flipped and the term continues to be false.

Now, let us look at the case where $T' \in \mathcal{T'}$. Since $\mathcal{B}(T')\ge 2$, choose $i'$ and $j'$ that corresponds to the chosen edge. Since $I$ is an independent set, at least one of these two endpoints $i'$ and $j'$ lies outside $I$. Without loss of generality let this endpoint be $j'$. But since variables outside $I$ are
fixed according to $x$, so the opposite literal of $T'$ on $j'$ is false.
Thus $T'$ is permanently false under the restriction.

Therefore the restriction computes $\mathrm{AND}_{|I|}$, up to input
negations. Approximate non-deterministic degree does not increase under restrictions and input negations. Therefore, we get,
\[
\mathsf{N}(f) \ge \mathsf{N}(\mathrm{AND}_{|I|}) \ge \Omega(\sqrt{|I|}) = \Omega\!\left(\sqrt{\frac{s_1(f)}{k+1}}\right).
\]
\end{proof}

\begin{lemma}[\cite{bafna2016sensitivity}]\label{lem:sens-width-lb}
Let $f$ be a Boolean function computed by a minimal read-$k$ DNF formula $C$ of maximum width $\beta$. Then,
\[ s_1(f) + (1+k)\,s_0(f) \ge \beta. \]
\end{lemma}

\begin{lemma}
Let $f:\{0,1\}^n\to\{0,1\}$ be computed by a minimal read-$k$ DNF
formula of maximum width $\beta$. Then
\[ (k+1)\left(\mathsf{N}(f)^2+ (2k+1)\mathsf{N}(\overline f)^2\right) \ge \Omega(\beta).\]
\end{lemma}

\begin{proof}
    By combining, Lemma~\ref{lem:read-k} with Lemma~\ref{lem:sens-width-lb} we get the desired result.
\end{proof}

\begin{theorem}\label{thm:readk-dnf-main}
Let $f:\{0,1\}^n\to\{0,1\}$ be a non-constant Boolean function computed by a minimal read-$k$ DNF formula with $\alpha$ terms and maximum width $\beta$. Then
\[ \widetilde{\deg}(f) \le O\!\left(k(k+1)^2(2k+1)^2\cdot \bigl(\mathsf{N}(f)\,\mathsf{N}(\overline f)\bigr)^6\right). \]
\end{theorem}
\begin{proof}
Notice that $C^{(1)}(f)\le \beta$ as any single satisfied term can certify $f=1$. Similarly, $C^{(0)}(f)\le \alpha$ as one falsified literal per term certifies $f=0$. 
Thus we get $\widetilde{\deg}(f)\le D(f)\le C^{(0)}(f)\,C^{(1)}(f)\le \alpha\beta$.

By Lemma~\ref{lem:sens-width-lb} (\cite{bafna2016sensitivity}),
$s_1(f)+(1+k)\,s_0(f)\ge\beta$. By Lemma~\ref{lem:read-k},
\[ s_1(f)\le O\bigl((k+1) \mathsf{N}(f)^2\bigr) \quad\text{and}\quad s_0(f)\le O\bigl((2k+1) \mathsf{N}(\overline{f})^2\bigr). \]

By Lemma~\ref{lem:ndeg-alpha} we know, 
\[ \alpha \beta \le O\bigl( k (\beta \mathsf{N}(\overline{f}))^2\bigr). \]

Also, since $f$ is non-constant, $\mathsf{N}(f),\mathsf{N}(\overline f)\ge 1$. 

Therefore, 
\begin{align*}
   \beta^2 &\le O\left((k+1)^2(N(f)^2+(2k+1)N(\overline{f})^2) ^2\right) \\ 
   &\le O\left((k+1)^2(2k+1)^2(N(f)^2N(\overline{f})^2)^2\right).
\end{align*}

So combining everything we get, 
\[ \widetilde{\deg}(f) \le O\bigl( k (\beta \mathsf{N}(\overline{f}))^2\bigr) \le O\bigl(k(k+1)^2(2k+1)^2(\mathsf{N}(f)\mathsf{N}(\overline{f}))^6\bigr). \]

\end{proof}

\subsection{Graph and \texorpdfstring{$k$}{k}-uniform hypergraph properties}
We use the construction from Lemma 1.2 from~\cite{chugh2022decisiontreecomplexityversus} to derive that approximate non-deterministic degree conjecture is true for Graph Properties. For easier readability, we repeat the proof in Appendix~\ref{Graph_property_appendix}.

\begin{corollary}[Lower bounds on degree and approximate non-deterministic degree of \texorpdfstring{$\mathcal{P}$}{P}]
\label{graph_property_result} Let $\mathcal{G}_n$ be the set of undirected graphs on $n$ vertices.
    For a graph property $\mathcal{P} :\mathcal{G}_n \rightarrow \{0,1\}$, \[\max\{\mathsf{N}(\mathcal{P})^6,\mathsf{N}(\overline{\mathcal{P}})^6\}= \Omega(n).\]
   
\end{corollary}

We now show that a similar bound can be proven for any Boolean function $\mathcal{P}$ that represents a $k$-uniform hypergraph properties.

\begin{theorem}[Uniform hypergraph properties]\label{thm:hypergraph}
Let $k\ge 1$ be a fixed constant, and let $\mathcal P:\{0,1\}^{\binom n k}\to\{0,1\}$ be a non-constant $k$-uniform hypergraph property on vertex set $[n]$.
Let $M(\mathcal P) = \max\{\apndeg(\mathcal P), \apndeg(\overline{\mathcal P})\}$.
Then
\[ M(\mathcal P) \ge \Omega(n^{1/6}). \]
Thus Conjecture~\ref{approx_ndeg_conjecture} holds for every fixed
$k$-uniform hypergraph property.
\end{theorem}
\begin{proof}
    When $k=1$ we obtain the symmetric case. So using Corollary~\ref{symmetric_relation} we get the result.

    Now we consider the case $k \geq 2$. Replacing $\mathcal P$ by $\overline{\mathcal P}$ does not change $D(\mathcal P)$ or $M(\mathcal P)$, so we may assume that the empty hypergraph does not satisfy $\mathcal P$.

    Let $H$ be a hypergraph satisfying $\mathcal P$ with the minimum possible number $m>0$ of hyper-edges. We split into two cases.

\begin{enumerate}
    \item Case $m > \frac{n}{2k}$: Restrict $\mathcal P$ by setting every hyper-edge outside $E(H)$ to $0$. On the remaining $m$ variables, the restricted function is $\mathrm{AND}_m$.
Note that any proper sub-hypergraph of $H$ has fewer than $m$ hyper-edges, and therefore cannot satisfy the property $\mathcal P$ by minimality of $H$. As the restrictions do not increase $M$, and since
$\apndeg(\mathrm{AND}_m)=\Omega(\sqrt m),$
we get
\[ M(\mathcal P) \ge M(\mathrm{AND}_m) = \Omega(\sqrt m) = \Omega(\sqrt {n/2k}). \]

Since $k$ is constant, we obtain $\Omega(\sqrt n)$ bound. 

\item Case $m \leq \frac{n}{2k}$:
Let $I$ be the set of isolated vertices (i.e., vertices with zero degree) of $H$. Since each hyperedge contains $k$ vertices, the number of non-isolated vertices is at most $km\le n/2$. Hence $|I|\ge n/2$. Now choose a vertex $v$ that lies in some hyper-edge of $H$. 

Let the link of $v$ consists of all $(k-1)$ sets obtained by removing $v$ from hyper-edges that contain $v$,
\[ L_v = \{S\in \binom{[n]\setminus\{v\}}{k-1}: S\cup\{v\}\in E(H)\}. \]

Since we are working with $k$-uniform hyper-graphs for $k\geq 2$, $L_v$ is not empty. Let $H' = H \setminus \{S\cup\{v\}: S\in L_v\}$.
Since $H'$ has fewer than $m$ hyper-edges, minimality of $H$ implies
$\mathcal P(H')=0$. Now define $I' := I\cup\{v\}$. So, $|I'|\ge n/2+1$. Also, no hyper-edge of $H'$ intersects $I'$, and every
$S\in L_v$ is disjoint from $I'$. Thus for every subset $A\subseteq I'$, define an $k$-uniform hypergraph $H_A$ by $E(H_A) = E(H') \cup \{S\cup\{u\}: S\in L_v,\ u\in A\}$.

Now we can embed a function on $I'$ variables such that the function is 1 iff we have the property $\mathcal P(H_A)$, i.e.,
\[ g:\{0,1\}^{|I'|}\to\{0,1\} \textit{\quad so that \quad }  g(1_A)=\mathcal P(H_A), \]
where $1_A$ is indicator function.
Note that this function $g$ is obtained from $\mathcal P$ by restricting variables so that the edges of $H'$ are fixed to $1$, and  all irrelevant edges are fixed to $0$, and for each $u\in I'$, all hyper-edges $S\cup\{u\}$ with $S\in L_v$ are identified with the variable of $I'$. By Lemma~\ref{lem:restrictions-projections} we get that $M(g)\le M(\mathcal P)$. Thus we need to lower bound $M(g)$.

Note that $g$ is non-constant: $g(0^{I'})= \mathcal{P}(H')=0$ and $g(1_{\{v\}}) = \mathcal{P}(H)=1$. 
We claim that this function $g$ is symmetric on $|I'|$. This is because every $u \in I'$ is isolate in $H'$ and for each $u \in I'$ the edges added in $H_A$ is exactly $\{S \cup \{u\} : S \in L_v\}$.  Hence all variables play identical roles and $g(1_A)$ is a function of $|A|$. Since $|I'|\ge \frac{n}{2}$+1, the symmetric case lower bound gives $M(g)=\Omega(|I'|^{1/6})=\Omega(n^{1/6})$. By Lemma~\ref{lem:restrictions-projections}, $M(g)\le M(\mathcal P)$.
Therefore $M(\mathcal P)\ge\Omega(n^{1/6})$.
\end{enumerate}
\end{proof}

\newpage

\bibliographystyle{alpha}
\bibliography{references}
\appendix

\section{Functions with bounded alternation number proof}
\label{alternating_number_appendix}

For convenience, we restate the theorem.
\begin{theorem*}
    [Restatement of Theorem~\ref{alternating_number_proof_2}] 
    Let $f:\{0,1\}^n \rightarrow \{0,1\}$ be alternation number $k$. Then, \[bs(f) \le O\left(k \max\{\mathsf{N}(f)^2,\mathsf{N}(\overline{f})^2\}\right).\]
\end{theorem*}
\begin{proof}
    Unlike zebra functions, here, there could be a monotone path $P$ with alternation number less than $k$. Any input $x \in \{0,1\}^n$ will satisfy at least one of the following three conditions:
    \begin{enumerate}
        \item $x \preceq u$ for some max term $u$.
        \item $x \succeq d$ for some min term $d$.
        \item There is no max term $u$ and min term $d$ such that $x \preceq u$ and $x \succeq d$.
    \end{enumerate}
         
     The third condition would be when there is a monotone path with no alternation. Also note that when the alternation number is odd, $f(0^n) \ne f(1^n)$, and hence every monotone path will have alternating number at least 1. This would mean that every $x \in \{0,1\}^n$ will follow either of the first two conditions and hence $bs(f,x) \le C(f,x) \le O(k \max\{\mathsf{N}(f)^2,\mathsf{N}(\overline{f})^2\})$ following same proof idea from Theorem~\ref{zebra}.

    Now, we only need to look at the case when $k$ is even, that is $k=2t$ for some $t$, and $x \in \{0,1\}^n$ is on a monotone path on which $f$ is constant which means $f$ is constant on $\{z \preceq x\}$ and $\{z \succeq x\}$. Therefore, $f(x)=f(0^n)=f(1^n)$. Let $bs(f,x)=\ell$ and $B_i$ with $i \in [\ell]$ be the minimal sensitive blocks of $x$. Therefore, $f(x^{B_i}) \ne f(x)$ and hence $f(x^{B_i}) \ne f(1^n)$ and $f(x^{B_i}) \ne f(0^n)$ which means $x^{B_i}$ follows both of the first two conditions and therefore, there exists $u$ such that $x^{B_i} \preceq u$ and there exists $d$ such that $d \preceq x^{B_i}$. Hence $bs(f,x^{B_i}) \le O(k  \max\{\mathsf{N}(f)^2,\mathsf{N}(\overline{f})^2\})$.
    
    Let $D_i =\{ j \in B_i \mid x_j=0 \}$ and $U_i = \{ j \in B_i \mid x_j=1\}$. $x^{U_i} \preceq x^{B_i} \preceq u$ and $d \preceq x^{B_i} \preceq x^{D_i}$. Since $B_i$ is the minimal sensitive block for $x$, $f(x)=f(x^{D_i})=f(x^{U_i})$. $U_i$ and $D_i$ cannot be empty. Lets say $U_i$ is empty then $D_i=B_i$ and hence $x \preceq x^{D_i}$ and $x^{D_i}=x^{B_i}$ and so $x \preceq x^{B_i}$ but this contradicts the assumption for third condition. Similarly, lets say $D_i$ is empty, then $U_i=B_i$ and hence $x^{B_i}=x^{U_i} \preceq x$ which would again contradict the assumption.  
    
    Let us now consider a fixed sensitive block $B_j$. The goal is to find the number of $i$'s such that $i\ne j$ and $f(x^{B_j})=f(x^{D_i \cup B_j})$ along with $f(x^{B_j})=f(x^{U_i \cup B_j})$ and then use it to bound the block sensitivity of $x^{U_i}$ or $x^{B_i}$ in terms of block sensitivity of $x$. Let $f'$ be the restriction of $f$ such that the zeros of $x^{B_j}$ are fixed and $f''$ is the restriction of $f$ such that $1's$ in $x^{B_j}$ are fixed and hence $f'$ is defined on the sub-cube $\{z: z \preceq x^{B_j}\}$ and $f''$ is defined on the sub-cube $\{z: z \succeq x^{B_j}\}$. For $i \ne j$, $x^{D_i \cup B_j}$ is obtained by taking $x^{B_j}$ and flipping its zeros at indices in $D_i$ to 1 and hence $x^{B_j} \preceq x^{D_i \cup B_j}$. This implies that $x^{D_i \cup B_j} \in \{z: z \succeq x^{B_j}\}$.
    For the function $f''$, since $C(f'',x^{B_j})$ is the certificate complexity on input $x^{B_j}$, there is an assignment $C': S \rightarrow \{0,1\}$ of size $C(f'',x^{B_j})$ that makes the function $f''$ on the sub-cube constant. We have,  
    \[ C(f'',x^{B_i}) \le O(\operatorname{alt}(f'')M(f'')^2) \le O(\operatorname{alt}(f'')M(f)^2). \]
    For every $D_i$ that does not intersect with $S$, $f(x^{D_i \cup B_j})=f(x^{B_j})$ because none of the variables corresponding to $S$ are flipped and therefore, $C'$ on $S$ still continues to certify the value of $f(x^{D_i \cup B_j})$ to be equal to $f(x^{B_j})$. The maximum number of $D_i$, with $i \ne j$ that can intersect with $S$ is $C(f'',x^{B_j})$ because all $D_i$ are disjoint. Let $\mathcal{D}_i = \left\{D_i \mid i \ne j \text{ and } f(x^{D_i})=f(x^{D_i \cup B_j})\right\}$. Therefore, \[
        \left|\mathcal{D}_i\right|\ge  \ell -1 - C(f'',x^{B_j}) \ge \ell -1 - O(\operatorname{alt}(f'')M(f)^2). 
    \] 

    Similarly, we can have same analysis for $f'$ and in this case we will look at $U_i$, with $i \ne j$ such that $f(x^{B_j})= f(x^{U_i \cup B_j})$. Let $\mathcal{U}_i = \{U_i \mid i \ne j \text{ and } f(x^{U_i})=f(x^{U_i \cup B_j})\}$
    Therefore, \[
        \left|\mathcal{U}_i\right|\ge  \ell -1 - C(f',x^{B_j}) \ge  \ell -1 - O(\operatorname{alt}(f')M(f)^2). 
    \] 
    But since $\operatorname{alt}(f')+\operatorname{alt}(f'')=k=2t$, we get,
    \begin{equation}\label{ui_di_bound}
        \left|\mathcal{U}_i\right|+\left|\mathcal{D}_i\right| \ge (2\ell -2 - O(tM(f)^2)).
    \end{equation}

    Now, this can also be expressed through a $2 \ell \times \ell$ matrix. Each row of this matrix corresponds to either $U_i$ or $D_i$ and each column corresponds to $B_j$. Let the entries of the matrix be as follows:
    \begin{itemize}
        \item The entries corresponding to row $U_i$ or $D_i$ and column $B_i$ are assigned to 1.
        \item The entries corresponding to row $U_i$ or $D_i$ and column $B_j$ with $i \ne j$ are assigned to 1 if $f(x^{T_i \cup B_j}) \ne f(x^{B_j})$ where $T_i$ is either $U_i$ or $D_i$.
        \item Rest of the entries are 0.
    \end{itemize}
    Therefore, total number of zeros in this matrix are at least $\ell (2\ell -2-O(tM(f)^2))$. We get this by multiplying the bound in equation~\ref{ui_di_bound} by $\ell$. Therefore, the number of ones are at most $ 2\ell +  \ell O(t  M(f)^2)$. On average, number of ones in each row is at most $1 + O(t M(f)^2)$. Lets take a row $T_i$ which has at most $1+O(t M(f)^2)$ ones. For this row, since the ones correspond to $f(x^{B_j})\ne f(x^{T_i \cup B_j})$, the number of entries with $f(x^{B_j})= f(x^{T_i \cup B_j})$ will be at least $\ell - 1 -O(tM(f)^2)$. From earlier relations, we also know that $f(x^{T_i}) = f(x)$ and $f(x) \ne f(x^{B_j})$. The number of disjoint sensitive blocks for $x^{T_i}$ will be at least $\ell - 1 -O(tM(f)^2)$. Additionally, since these entries did not include $B_i$, and $x^{T_i}$ is also sensitive to $B_i \backslash T_i$, we need to add 1 to earlier number. Therefore, we have a lower bound on block sensitivity of $bs(f,x^{T_i})$.
    \[ bs(f,x^{T_i}) \ge \ell -1 -O(tM(f)^2) +1 \ge \ell - O(t M(f)^2).  \]
    But since $x^{T_i}$ satisfies one of the first two conditions, we have an upper bound on $bs(f,x^{T_i})$.
   \[ bs(f,x^{T_i}) \le O(k M(f)^2). \] 

   Putting it all together, since $bs(f,x)=\ell$, we get,
   \[ bs(f,x) - O\left(\frac{k}{2} M(f)^2\right) \le O(k M(f)^2). \]
   Therefore, $bs(f,x) \le O(k M(f)^2)$ for all $x \in \{0,1\}^n$.
\end{proof}

\section{Graph Property Proof}
\label{Graph_property_appendix}

\begin{corollary*}
    [Restatement of Corollary~\ref{graph_property_result}] Let $\mathcal{G}_n$ be the set of undirected graphs on $n$ vertices.
    For any non-constant graph property $\mathcal{P} :\mathcal{G}_n \rightarrow \{0,1\}$, \[\max\{\mathsf{N}(\mathcal{P})^6,\mathsf{N}(\overline{\mathcal{P}})^6\} \ge \Omega(n).\]
   
\end{corollary*}

 \begin{proof}
        Assume that empty graph is not in $\mathcal{P}$. Now, the proof will be done case by case on the least number of edges $m>0$ in any graph in $\mathcal{P}$. Taking an arbitrary graph $G=(V,E) \in \mathcal{P}$ with $m$ edges, and representing it as a string $x_G \in \{0,1\}^{\binom{n}{2}}$ such that $(x_G)_{\{u,v\}}=1$ if and only if there is an edge between vertex $u$ and $v$. 
        \begin{enumerate}
            \item Case $m > \frac{n}{4}$:
            By taking the restriction of $\mathcal{P}$ where $x_{\{u,v\}}=0$ if $\{u,v\} \notin E$ and noting that if any of the other variables (these would be $m$ variables) are made 0, then that graph can never be in this restriction of $\mathcal{P}$ since then it would violate the minimality condition. So the restriction now becomes an AND function on the $m$ unset variables. 
            We know that approximate non-deterministic degree of AND function on $m$ variables is $\Omega(\sqrt{m})$ and also that approximate non-deterministic degree does not increase under restriction. Therefore, $\mathsf{N}(\mathcal{P})> \Omega(\sqrt{n/4})$.
            \item Case $m \le \frac{n}{4}$:
            Since the number of edges of $G$ is less than $n/4$, the number of vertices with zero degree is at least $n/2$. Let there be $k$ vertices $v_1,v_2,\dots,v_k$ with non-zero degrees and $n-k$ vertices $v_{k+1},\dots,v_n$ with degree zero. Now, lets say the vertex $v_k$ is connected to $v_1,\dots,v_d$. If all these edges are removed, then the graph that is obtained, let that graph be $G'$. $G'$ clearly has less than minimum edges $m$ to be in $\mathcal{P}$ and therefore, 
            \[G' \notin \mathcal{P}.\]
            Now, let $G_j$ be the graph obtained by adding the edges $\{v_j,v_1\},\{v_j,v_2\},\dots,\{v_j,v_n\}$ to graph $G'$ where $j \in \{k,k+1,\dots,n\}$. All these graphs are isomorphic to $G$ and so $G_j \in \mathcal{P}$. Also, $G_k$ and $G$ are the same.
            Now, let $f: \{0,1\}^{n-k+1} \rightarrow \{0,1\}$ and $x=(x_k,x_{k+1},\dots,x_n) \in \{0,1\}^{n-k+1}$. Define graph $G_x$ on $n$ vertices such that $f(x)=\mathcal{P}(G_x)$. $G_x$ contains edges of $G'$ and $\{v_i,v_1\},\{v_i,v_2\},\dots,\{v_i,v_d\}$ for each $i \in \{k,\dots\,n\}$ such that $x_i=1$.
            Let us first see that $f$ is a non-constant function.
            Let us look at the value of $f(0^{n-k+1})$ and since in this case $x=(0,\dots,0)$, the only edges in $G_{0^{n-k+1}}$ is the edges of $G'$ and therefore, $\mathcal{P}(G_{0^{n-k+1}})=0$ and $f(0^{n-k+1})=0$. For a string $|x|=1$, there exists some $j \in \{k,k+1,\dots,n\}$ such that $x_j=1$ and then $G_x = G_j$ and therefore for such $x$, $f(x)=1$ since $G_j \in \mathcal{P}$ as seen earlier. And thus, $f$ is a non-constant function.
            Now, lets see that $f$ is a symmetric function. For any $x,x'$, such that $|x|=|x'|$, $G_x$ and $G_{x'}$ are isomorphic and since $\mathcal{P}$ is graph property, $f(x)=\mathcal{P}(G_x)=\mathcal{P}(G_{x'})=f(x')$.
            We already know from Corollary~\ref{symmetric_relation} that $\deg(f) \le \max\{\mathsf{N}(f)^6,\mathsf{N}(\overline{f})^6\}$
            and from~\cite{von1997polynomials} that degree for a non-constant symmetric Boolean function is linear in $n-k+1$ and since $k \le n/2$, $\deg(f)=\Omega(n-k+1)=\Omega(n)$.

            $f$ is obtained from $\mathcal{P}$ by restriction and identification of variables: Every variable corresponding to the pair $\{v_i,v_j\}$ such that $1 \le i<j \le k-1$ and $\{i,j\}$ is not an edge in $G'$ is set to 0. For each $i \in \{k,k+1,\dots,n\}$, the variables corresponding to the pairs $\{v_i,v_1\},\dots,\{v_i,v_d\}$ are identified.
            Since approximate non-deterministic degree does not increase under restriction and identification of variables, therefore, $\max\{\mathsf{N}(\mathcal{P})^6,\mathsf{N}(\overline{\mathcal{P}})^6\} \ge \Omega(n)$.
        \end{enumerate}
        Taking into account, both the above case, we have 
        \[\max\{\mathsf{N}(\mathcal{P})^6,\mathsf{N}(\overline{\mathcal{P}})^6\} \ge \Omega(n).\]
    \end{proof}

\end{document}